\documentclass[12pt]{article}
\usepackage{amsmath, amsthm, amssymb, bbm, setspace,bigints}
\usepackage[margin=1 in]{geometry}
\usepackage{caption, enumitem}
\usepackage{subcaption}
\usepackage[toc,page]{appendix}
\usepackage{graphicx, amsfonts, tikz, multirow}
\usepackage{algorithm, algpseudocode, bm, nicematrix}
\usepackage{adjustbox}
\usepackage{natbib}
\usepackage{mathtools}
\usetikzlibrary{bayesnet}
\usepackage{pdfpages}
\usepackage{epstopdf, pdflscape}
\usepackage{booktabs}
\usepackage[colorlinks,citecolor=blue,urlcolor=blue]{hyperref}
\usepackage[utf8]{inputenc}
\usepackage{authblk}
\usepackage{comment}
\usepackage{xr}

\def\Gauss{{ \mathrm{N} }}
\def\T{\mathrm{\scriptscriptstyle{T}}}

\def\ind{\mathbbm{1}}
\newcommand{\supp}{\mathrm{Supp}}
\newcommand{\norm}[1]{\Vert#1\Vert}

\allowdisplaybreaks
\newtheorem{theorem}{Theorem}[section]
\newtheorem{lemma}[theorem]{Lemma}
\newtheorem{remark}{Remark}

\newtheorem{example}{Example}

\newtheorem{proposition}[theorem]{Proposition}
\newtheorem{corollary}[theorem]{Corollary}

\newtheorem{assumption}[theorem]{Assumption}

\DeclareMathOperator{\tr}{tr}

\newcommand{\Cov}{\mathrm{Cov}}

\usepackage{mathtools}
\usepackage{cleveref}

\date{}
\title{Empirical Bayes data
integreation for multi-response regression}
\author[1]{Antik Chakraborty\thanks{antik015@purdue.edu}}
\author[1]{Fei Xue \thanks{feixue@purdue.edu}}
\affil[1]{Department of Statistics, Purdue University}
\begin{document}
\maketitle
\begin{abstract}
    Motivated by applications in tissue-wide association studies (TWAS), we develop a flexible and theoretically grounded empirical Bayes approach for integrating 
data obtained from different sources. We propose a linear shrinkage estimator that effectively shrinks singular values of a data matrix. This problem is closely connected to estimating covariance matrices under a specific loss, for which we develop asymptotically optimal estimators. The basic linear shrinkage estimator is then extended to a local linear shrinkage estimator, offering greater flexibility. Crucially, the proposed method works under sparse/dense or low-rank/non low-rank parameter settings unlike well-known sparse or reduced rank estimators in the literature. Furthermore, the empirical Bayes approach offers greater scalability in computation compared to intensive full Bayes procedures. The method is evaluated through an extensive set of numerical experiments, and applied to a real TWAS data obtained from the Genotype-Tissue Expression (GTEx) project.
\end{abstract}

\section{Introduction}\label{sec:intro}
Genome-wide association studies (GWAS) aim to identify potential genotype markers associated with a particular phenotype, typically a disease. GWAS data usually involve a large number of genes, which caused the explosion of statistical methods that are able to handle many variables. 
More recently, genetic scientists are collecting gene expression data from multiple tissues, e.g., brain tissues and heart tissues \citep{mai2023transcriptome, xue2022empirical}. 
To leverage these multi-tissue gene expression data in identification of genotype markers, tissue-wide association studies (TWAS) 
prioritize genes that are functionally linked to the phenotype by associating genetically predicted gene expression with the phenotype \citep{wainberg2019opportunities}, where we need to predict multi-tissue gene expression values based on genotype data.

Although individual analysis on the prediction of gene expression in each tissue is possible with existing methods, the resulting analysis does not integrate potential shared information across tissues. Indeed, when data from multiple sources with some commonality are available, a joint analysis across all the data sources allows for borrowing of strength. 
{
However, data obtained from TWAS studies might
not adhere to 
sparsity or low-rank structure \citep{heap2009complex, gresle2020multiple}, which is a crucial assumption for many of the available statistical methods 
\citep{velu2013multivariate}.}
Such assumptions are also very hard to verify in practice.
In this article, our aim is to develop methods that a) successfully integrate data across multiple sources (e.g., multiple tissues), b) are computationally scalable, and c) perform well regardless of specific structures within the parameter. 

Specifically, we focus on the case where the ordinary least squares (OLS) estimate is available under a linear regression model with a response and a set of predictors in each data source (or tissue).
Across different data sources,
the predictors are the same but the response varies. 
Vector-valued outcomes from multiple sources are also encountered in other scientific disciplines including finance, bioinformatics, and growth curve models. For example, in finance stock prices of multiple companies are studied in relation to the same set of predictors. 

Under the assumption of a linear relationship between  the predictors and the response variables in all data sources, the interest centers on recovering the matrix of regression coefficients. 
Traditionally, this problem was studied through the lens of reduced-rank regression \citep{anderson1951estimating, izenman1975reduced, velu2013multivariate, geweke1991efficient}. More recently, methods that are able to handle high-dimensional (many predictors) along with the reduced rank nature of the coefficient matrix have been developed \citep{yuan2007dimension, bunea2011optimal, bunea2012joint, chen2012sparse} catering to modern applications, including denoising Gaussian matrices, a closely related problem. Within the frequentist framework, the reduced rank constraint and presence of effect of a subset of predictor variables are most naturally expressed in terms of penalized regression which can be interpreted as suitable priors over the parameter space leading to a maximum {\it a posteriori} (MAP) interpretation of the estimators. Full Bayes treatment of the problem has also been carried out \citep{bai2018high, chakraborty2020bayesian}. The resulting procedures principally shrink the coefficient matrix towards low-rank structures. 

A common thread between all penalty-based methods and Bayesian versions thereof is the assumption of an underlying structure (low-rank and/or row-sparse) in the coefficient matrix, which statistically is meaningful but is very hard to verify in practice. Here, we take a different view to the problem in line with the three objectives outlined earlier. Our solution is an empirical Bayes one, which is able to borrow information from multiple sources, avoids computationally intensive full Bayes procedures, and applies to situations where no specific structural information about the parameter is available.
In fact, the OLS estimates from each source can be treated as observed data under an additive model where Gaussian noise is added to the true matrix of coefficients. 
\cite{efron1972empirical} proposed an empirical Bayes estimator for this mean matrix estimation problem 
which induces linear shrinkage on the singular values of the observation matrix. More recently, \cite{matsuda2015singular} developed superharmonic priors for singular value shrinkage for matrix-valued mean parameters. In their development, they closely follow \cite{stein1981estimation} who showed optimality properties of Bayes estimates with a superharmonic prior distribution. These singular value superharmonic priors were then used in the context of matrix completion \citep{matsuda2019empirical}, estimation under matrix quadratic loss \cite{matsuda2022estimation}. This class of priors also place increasing amount of mass near low-rank matrices, thus implicitly assuming such an underlying structure.
Moreover, \cite{wang2021linear} developed an empirical Bayes estimator for multivariate linear regression problems but they mainly focused on prediction.

\cite{efron1972empirical, efron1976multivariate} noted that the linear shrinkage estimator of the mean matrix can alternatively be interpreted as the posterior mean under Gaussian priors; the optimal decision in this context under the Frobenius loss. However, from a practical perspective, specification of Gaussian priors requires one to specify a prior covariance matrix which is not immediate, especially without certain structural assumption on the parameter. 
\cite{efron1976multivariate} subsequently show that the success of linear shrinkage estimators relies on the accurate estimation of the marginal covariance of the data under the relative savings loss for estimating covarince matrices. \cite{efron1976multivariate} considered rotation invariant estimators of the covariance matrix.
Their initial suggestion was to use a linear shrinkage of eigenvalues for the covariance estimation problem. Linear shrinkage of eigenvalues was pioneered by \cite{ledoit2004well}. In a series of papers the authors have developed general non-linear shrinkage estimators and studied their properties; see for example \cite{ledoit2011eigenvectors, ledoit2012nonlinear, ledoit2018optimal, ledoit2022quadratic}. Although \cite{ledoit2022quadratic} established an optimal (asymptotic) shrinkage rule under several loss functions within the class of rotation invariant estimators, they have not considered the relative savings loss which is important for the regression coefficient matrix estimation. {Coincidentally, the need to accurately estimate the covariance matrix is also necessary for prediction problems. \cite{banerjee2021improved} studies Bayes predictive estimators for multivariate Gaussian models. Here, the optimal Bayes rule involves a quadratic form in the unknown covariance. The authors assume a spiked covariance structure \citep{paul2007asymptotics} for estimation of the unknown covariance, whereas, in this work, we take a loss-minimization-based approach.} 

In this paper, we consider the regression coefficient matrix estimation problem under the empirical Bayes framework, which mainly relies on the estimation of the covariance matrix of standardized OLS estimates.
We first develop an asymptotically optimal shrinkage rule for estimating the covariance matrix under the relative savings loss, and then propose a linear shrinkage estimator of the regression coefficient matrix based on the estimated covariance matrix. In this way, our proposed coefficient estimator is optimal asymptotically within the class of linear shrinkage estimators of the regression coefficient matrix. 
The proposed estimators are derived under settings when the number of data sources available is larger than the number of variables, and when it is not.

The proposed shrinkage rule that minimizes the relative savings loss is defined in terms of a smoothing parameter.
Our next contribution is to develop a data-dependent choice of this smoothing parameter,
using the technique of unbiased risk estimation (SURE). Our numerical experiments reveal that such data-dependent tuning often results in improved risk results. 

Finally, to gain further flexibility, we extend our estimator to a situation where the prior is a mixture of Gaussian, which results in an adaptively weighted local linear shrinkage rule for estimating the mean matrix. This is useful for capturing any complex prior structure in the parameter. However, a fundamental benefit of the proposed approach is that we can compute the proposed estimator without having to carefully devise an explicit prior that embeds this complex structure. Our numerical experiments reveal that the local linear shrinkage estimator has better or at par performance with estimators that are specifically designed for structural parameters, without resorting to such assumptions.

In Section \ref{sec:estimator} we introduce the problem and propose the estimator. Section \ref{sec:covariance_shrinkage} is devoted to the development of the covariance estimator and its data-dependent version. Section \ref{sec:local_linear_shrinkage} describes a version of the proposed estimator under a mixture prior. 
In Sections \ref{sec:experiments} and \ref{sec:real_data}
we 
evaluate the proposed estimator through numerical experiments, compare it with other approaches, and apply it to the  Genotype-Tissue Expression (GTEx) data \citep{lonsdale2013genotype}.

\section{Data integration by linear shrinkage}\label{sec:estimator}

{
In this section, we introduce the problem setup and propose an empirical Bayesian framework for estimation of genotype effects on gene expression levels.}
For a certain gene, we let $y_i^{(t)}$ be its expression level for the $i$-th subject in the $t$-th tissue, and $x_i$ be a fixed $p$-dimensional vector of single nucleotide polymorphisms (SNPs) for the $i$-th subject. Here $i = 1, \ldots, N$ and $t = 1, \ldots, n$. We assume the sample size $N>p$ but allow the tissue size $n$ to be either larger or smaller than $p$. Consider a linear regression model for the $t$-th tissue 
\begin{equation}\label{eq:tissue_linear_regression}
    y_i^{(t)}= x_i^\T \beta^{(t)} + \varepsilon_i^{(t)},
\end{equation}
where $\beta^{(t)}$ is a $p$-dimensional coefficient vector and $\varepsilon_i^{(t)}\sim N(0, \sigma^2)$ is the error term. We assume that $\varepsilon_1^{(t)}, \dots, \varepsilon_N^{(t)}$ are conditionally independent given the tissue-specific coefficient $\beta^{(t)}$. 

The OLS estimator $\widehat{\beta}^{(t)}=(X^\T X)^{-1}X^T y^{(t)}$ is a single-tissue estimator of $\beta^{(t)}$, where $X=(x_1, \dots, x_N)^T$ is the $N\times p$ fixed design matrix, and $y^{(t)}$ is a vector consisting of  $y_i^{(t)}$ ($1\le i \le N$). 
Since $\varepsilon_i^{(t)}\sim N(0, \sigma^2)$, we have $\widehat{\beta}^{(t)} \mid \beta^{(t)} \sim \Gauss(\beta^{(t)}, \sigma^2 (X^\T X)^{-1})$.
Let $B$ be a $n \times p$ matrix with $\beta^{(t)}$ as the $t$-th row, and $\widehat{B}$ be a $n \times p$ matrix with $\widehat{\beta}^{(t)}$ as the $t$-th row. 

By sufficiency of the OLS estimator, we can consider $\widehat{B}$ as our data matrix and aim to estimate $B$ under the model
\begin{equation}\label{eq:base_model}
    \widehat{ \beta}^{(t)} = {\beta^{(t)}} + {u}^{(t)} \ \  \text{with} \ \ {u^{(t)}} \sim \Gauss(0, Q), 
\end{equation}
 where $Q = \sigma^2(X^\T X)^{-1}$. 
To leverage shared information across multiple tissues, we propose to assign a common prior $\pi$ to $\beta^{(t)}$ for $t=1,\dots, n$. 
 For an estimator $\widetilde{B}$ 
 of $B$, we consider the Frobenius loss
\begin{align}\label{eq:general_squared_loss}
    \mathbf{L}(B, \widetilde{B})= \sum_{t=1}^n (\widetilde{\beta}^{(t)} - \beta^{(t)})^T (\widetilde{\beta}^{(t)} - \beta^{(t)}) = \tr [(\widetilde{B}- B)(\widetilde{B}-B)^\T]
\end{align}
where for any matrix $A$, we write $\mathrm{tr}(A) = \sum_{j} A_{jj}$. The corresponding posterior expected loss is $\mathbb{E}_{B\mid \widehat{B}}[\mathbf{L}(B, \widetilde{B})]$, and the minimizer of the $\mathbb{E}_{B\mid \widehat{B}}[\mathbf{L}(B, \widetilde{B})]$ is the posterior expectation $\mathbb{E}(B \mid \widehat{B})$. This is a vector-valued version of the canonical Normal means problem, which was considered by \cite{efron1972empirical} as an extension of Stein's shrinkage idea to vector-valued observations. In the most general setup of the problem, $\sigma^2$ is unknown. However, it can be estimated from the tissue-specific regressions, and in this paper we use the average of those estimators as a fixed value of $\sigma^2$.







In particular, 
a Gaussian prior distribution for $\beta^{(t)}$ yields a linear shrinkage decision rule.
Specifically, 
if $\pi$ is $\Gauss(0, Q^{1/2}\Omega Q^{1/2})$, then $\mathbb{E}(\beta^{(t)} \mid \widehat{B})=(\mathrm{I} - C) \widehat{\beta}^{(t)}$, where $C = Q^{1/2}(\mathrm{I} +\Omega)^{-1} Q^{-1/2}$. 
Here, we scale the prior with the observation noise $Q$ which is quite common in the Bayesian literature \citep{park2008bayesian}. In order to use this estimator, one needs to specify $\Omega$, which is not straightforward without making structural assumptions on $\beta^{(t)}$. However, an empirical Bayes analysis could still be carried out without this specification by noting that $\Sigma = (\mathrm{I} + \Omega) = \Cov(\widehat{\beta}^{(t)}_\star)$, where $\widehat{\beta}^{(t)}_\star = Q^{-1/2} \widehat{\beta}^{(t)}$. That is, we can estimate $\Omega$ or $\Sigma$ from the observed data.

Therefore, given an estimate $\widehat{\Sigma}^{-1}$ of $\Sigma^{-1}$, 
we propose to estimate the parameter $\beta^{(t)}$ using an estimated posterior mean  $\mathbb{E}(\beta^{(t)} \mid \widehat{B}) = (\mathrm{I} - \widehat{C})\widehat{\beta}^{(t)}$, where $\widehat{C} = Q^{1/2}\widehat{\Sigma}^{-1}Q^{-1/2}$ is
a plug-in estimate of $C$ since $Q$ is known. 
This implies that there is an intricate connection between estimating the parameter matrix $B$ under the Frobenius loss and estimating the marginal covariance matrix $\Sigma$.
In fact, by \cite{efron1976multivariate}, estimating $B$ under Frobenius loss, within the class of linear shrinkage estimators, i.e. {$\widetilde{\beta}^{(t)} = (\mathrm{I} - \widehat{C})\widehat{\beta}^{(t)}$}, is equivalent to the problem of estimating $\Sigma^{-1}$ under 
a relative savings loss $\mathbf{L}(\Sigma^{-1}, \widehat{\Sigma}^{-1}) = \mathrm{tr}((\Sigma^{-1} - \widehat{\Sigma}^{-1})^2 S)$, where $S =\sum_{t=1}^n \widehat{\beta}^{(t)}_\star \widehat{\beta}_\star^{{(t)}^\T}$. {This equivalence is recorded in the following Proposition \ref{prop:equivalence} for the sake of completeness. Its corresponding proof is provided in Section S.6.1 
of the Supplement.
\begin{proposition}\label{prop:equivalence}
    Suppose $\beta^{(t)} \sim \Gauss(0, Q^{1/2}\Omega Q^{1/2})$ with known $Q$ and $\hat{\beta}^{(t)} \mid \beta^{(t)} \sim \Gauss(\beta^{(t)}, Q)$. Consider the loss $\mathbf{L}(B, \widetilde{B})= \tr [(\widetilde{B}- B)(\widetilde{B}-B)^\T]$ and estimators $\widetilde{B} = \widehat{B}(\mathrm{I} - \widetilde{C})$, where $\widetilde{C}$ is an estimator of $C = Q^{1/2}(\mathrm{I} + \Omega)^{-1} Q^{-1/2}$. Then 
    $$\mathbb{E}_{(B, \widetilde{B})}[\mathbf{L}(B, \widetilde{B})] = \mathbb{E}_{\widehat{B}}[\tr \{(\widetilde{\Sigma}^{-1} - \Sigma^{-1})^2 \}\widehat{B}^\T \widehat{B}] +\text{constant}. $$
\end{proposition}
}

Hence, the best linear shrinkage estimator of $B$ should be based on the estimate $\widehat{\Sigma}^{-1}$ that is optimal in terms of the relative savings loss. 
We now turn our attention to the problem of estimating 
$\Sigma^{-1}$
under the loss $\mathbf{L}(\Sigma^{-1}, \widehat{\Sigma}^{-1}) = \mathrm{tr}((\Sigma^{-1} - \widehat{\Sigma}^{-1})^2 S)$.
\cite{efron1972empirical, efron1976multivariate} suggested two estimators for $\widehat{\Sigma}^{-1}$. The first is the natural unbiased estimator which is obtained by observing that $\widehat{\beta}^{(t)}_\star \overset{iid}{\sim} \Gauss(0, \Sigma)$, and the standard multivariate Gaussian distribution theory yields $S^{-1} \sim \text{inv-Wishart}(\Sigma^{-1}, n)$. Thus, $\widehat{\Sigma}^{-1}=(n-p-1)S^{-1}$ is an unbiased estimator of $\Sigma^{-1}$. The second involves a linear shrinkage estimator of $\Sigma^{-1}$. But linear shrinkage might not perform well under certain situations \citep{ledoit2012nonlinear}.

\section{Covariance shrinkage}\label{sec:covariance_shrinkage}

In this section, we develop a rotation invariant estimator for $\Sigma$.
Since this is an independent problem of interest, we consider a general setup: suppose that $n$ $p$-dimensional independent and identically distributed observations are available with zero mean and covariance matrix $\Sigma_n=\Sigma$. Our following results include cases when $n >p$ and $n<p$. In this section, we use the subscript $n$ to emphasize on the asymptotic framework we work in. 
The observed data is arranged in an $n \times p$ matrix $Z_n$. In the notation of the previous section, the rows of $Z_n$ are given by $\widehat{\beta}^{(t)}_\star$. Let $S_n = n^{-1}Z_n^\T Z_n$ be the sample covariance matrix. Consider the spectral decomposition of $S_n = U_n \Lambda_n U_n^\T = \sum_{i=1}^p \lambda_{n,i} u_{n,i}u_{n,i}^\T$, where $U_n$ is an orthogonal matrix with $u_{n,i}$ as the $i$-th column of $U_n$ and $\Lambda_n$ is a diagonal matrix with elements $ \bm{\lambda}_n= (\lambda_{n,1}, \ldots, \lambda_{n, p})^\T$ as the corresponding eigenvalues arranged in a non-decreasing order. 

We focus on the class of rotation invariant estimators $\widetilde{\Sigma}_n = U_n \widetilde{\Delta}_n U_n^\T$, where $\widetilde{\Delta}_n = \text{diag}(\delta_{n}(\lambda_{n,1}), \ldots, \delta_{n}(\lambda_{n,p}))$ and $\delta_n$ is a  {positive} univariate function that may depend on $S_n$. 
Such estimators $\widetilde{\Sigma}_n$ are 
rotation invariant since multiplying the data $Z_n$ by an orthogonal matrix with a determinant of one rotates the estimators accordingly.
Modern high-dimensional methods often rely on low-dimensional structures of the data when considering the problem of estimating unknown covariance matrices, for instance, sparsity. While properties of these estimators are well understood, relatively little is known about their performance when such assumptions do not hold. Instead, we pursue an estimator that principally shrinks eigenvalues without assuming any structure.
As indicated earlier, we focus on the problem of estimation of $\Sigma_n^{-1}$ under the loss $\mathbf{L}_n(\Sigma_n^{-1}, \widetilde{\Sigma}_n^{-1}) = \mathrm{tr}[(\Sigma_n^{-1} - \widetilde{\Sigma}_n^{-1})^2 S_n]$.

In our study, we consider a general version of the relative savings loss. Specifically, we 
consider 
$$\mathbf{L}_{m,n}(\Sigma_n^{-1}, \widetilde{\Sigma}_n^{-1}) 
=  \frac{1}{p} \mathrm{tr}[(\Sigma_n^{-1} - \widetilde{\Sigma}_n^{-1})^2 S_n^m ]$$ 
for $m = 0,1, 2, \dots$.
{\begin{remark}
    The case $m = 0$ corresponds to the inverse Frobenius loss; see also \cite{ledoit2018optimal, haff1979estimation}, whereas for $m=1$, we recover the relative savings loss, the focus of our paper. \cite{boukehil2021estimation, kubokawa2008estimation} studied the case $m = 2$.
    From a practical perspective, large values of $m$ put increasingly larger weight on large sample eigenvalues in terms of their contribution to the loss. Another major motivation for studying this general class of loss functions is to draw a distinction between the cases $m = 0$ and $m\geq 1$. As it turns out, if $m = 0$, then the optimal shrinkage function depends on the population eigenvalues through the limiting population eigenvalue distribution $H$. This problem was addressed in \cite{ledoit2018optimal}. Their solution was a numerical one, namely, the QuEST function. However, our results will show that for $m\geq 1$,  the optimal solution does not involve $H$. In other words, an explicit solution is available.
\end{remark}
}
Our goal here is to provide an optimal shrinkage rule under the loss function $\mathbf{L}_{m,n}$. We do so using three key steps - 1) we first find an almost sure non-random limit of $\mathbf{L}_{m,n}$, 2) then we find the shrinkage rule which minimizes this limit, and 3) find a consistent estimator of this optimum shrinkage rule. 

To find the almost sure limit of $\mathbf{L}_{m,n}$, we introduce the following notations and re-write the loss function.
We let $F_n(x) = \frac{1}{p} \sum_{j=1}^p \ind_{\{\lambda_{n,j} \leq x\}}$ and  $F_n^*(x) = p^{-1} \sum_{j=1}^p \ind_{\{\lambda_{n,j}^{-1} \leq x\}}$ be the empirical distribution functions of the sample eigenvalues and inverse eigenvalues, respectively. Clearly, for $x>0$, $F_n^*(x) = 1 - F_n(1/x)$. Recall that $\widetilde{\Sigma}_n = U_n \widetilde{\Delta}_n U_n^\T$, where $\widetilde{\Delta}_n = \text{diag}(\delta_{n}(\lambda_{n,1}), \ldots, \delta_{n}(\lambda_{n,p}))$.
Let $\delta_{n, j}=\delta_{n}(\lambda_{n,j})$ for $j=1,\dots, p$.
We then have
\begin{align*}
    \mathbf{L}_{m,n}(\Sigma_n^{-1}, \widetilde{\Sigma}_n^{-1}) 
    = \frac{1}{p} \mathrm{tr}[(\Sigma_n^{-1} &- \widetilde{\Sigma}_n^{-1}) (\Sigma_n^{-1} - \widetilde{\Sigma}_n^{-1}) S_n^{m}]\\
    = \frac{1}{p} \sum_{j=1}^p (u_{n,j}^\T \Sigma_n^{-2} u_{n,j}) \lambda_{n,j}^m & - \frac{2}{p}\sum_{j=1}^p (u_{n,j}^\T \Sigma_n^{-1}u_{n,j}) \frac{\lambda_{n,j}^m}{\delta_{n,j}} + \frac{1}{p} \sum_{j=1}^p \frac{\lambda_{n,j}^m}{\delta_{n,j}^2} \\
     =  \int_{-\infty}^{\infty} x^m d\Phi_n^{(-2)} (x) & - 2 \int_{-\infty}^{\infty} \frac{x^m}{\delta_{n}(x)} d\Phi_{n}^{(-1)}(x) + \int_{-\infty}^{\infty} \frac{x^m}{\delta^2_n(x)} dF_n(x),
\end{align*}
where $\Phi_n^{(-l)}(x) = \frac{1}{p} \sum_{j=1}^p (u_{n,j}^T\Sigma_n^{-l}u_{n,j})\ind_{[\lambda_{n,j}, \infty)}(x)$ for $l=1$ and $2$. 
Let $\Sigma_n = V_n \Gamma_n V_n^\T$ be the spectral decomposition of the true covariance matrix $\Sigma_n$, $v_{n,k}$ be the $k$-th column of $V_n$, and $\gamma_{n,k}$ be the $k$-th diagonal element of $\Gamma_n$. 
Note that, for $l = 2$, we have $\Phi_n^{(-2)}(x) = \frac{1}{p} \sum_{i=1}^p \ind_{[\lambda_{n,i}, \infty)}(x) \sum_{k=1}^p |u_{n,k}^\T v_{n,k}|^2 \gamma_{n,k}^{-2}$. 
We also let $H_n(x) = p^{-1} \sum_{j=1}^p \ind_{\{\gamma_{n,j} \leq x\}}$ be the empirical distribution of population eigenvalues. 

We make the following assumptions on the data and population distribution.
\begin{assumption}(Dimension)\label{assup_ratio}
    The concentration ratio $p/n\to c$ as $n\to \infty$, where $c>0$. We consider two scenarios:
    \begin{enumerate}[label=(\alph*)]
        \item The concentration ratio $c<1$, and there is a compact interval in $(0,1)$ that contains $p/n$ for large $n$.
        \item The concentration ratio $c>1$, and there is a compact interval in $(1,\infty)$ that contains $p/n$ for large $n$.
    \end{enumerate}
\end{assumption}

\begin{assumption}(Population)\label{assup_sigma}\leavevmode
\begin{enumerate}[label=(\alph*)]
    \item  The population covariance matrix $\Sigma_n$ is a $p\times p$ nonrandom symmetric positive-definite matrix. 
    \item  We assume that $H_n$ converges weakly to a limiting spectral distribution $H$, whose support, denoted by $\supp (H)$, is a finite union of closed intervals away from zero.
    \item  There exists a closed interval in $(0, +\infty)$ that contains $\gamma_{n,1}, \dots, \gamma_{n,p}$ for all large $n$.
\end{enumerate}
\end{assumption}

\begin{assumption}(Data)\label{assup_iid}
    The observed matrix $Z_n=W_n\Sigma_n^{1/2}$, where $W_n$ is a $n\times p$ matrix of independent and identically distributed (i.i.d) random variables with mean zero, variance one, and a finite $12$th moment.
\end{assumption}


Following \cite{silverstein1995analysis}, \cite{silverstein1995empirical}, and \cite{silverstein1995strong}, under Assumptions \ref{assup_ratio}(a), \ref{assup_sigma}, \ref{assup_iid}, we have $F_n(x) \overset{a.s.}{\to} F(x)$ as $n\to \infty$ for any $x\in\mathbb{R}$, where $F(x)$ is a continuously differentiable limiting spectral distribution function. This in turn implies the strong convergence of $F_n^*(x)$ to $F^*(x)$ where $F^*(x) = 1 - F(1/x), \, x>0$. Moreover, by \cite[Theorem 1.1]{bai1998no}, under the same assumptions, there exists a finite number $K\ge 1$ such that $\supp (F)=\cup_{k=1}^K [a_k, b_k]$, where $\supp (F)$ denotes the support set of $F$ and $0< a_1 < b_1<a_2 < b_2< \ldots < a_K < b_K < \infty$. Since $H_n$ also has a limit $H$, the limiting spectral distribution $F_n$ of sample eigenvalues which is $F$, is uniquely determined by the concentration ratio and $H$.

\begin{assumption}[Limiting shrinkage function]\label{assup_est_eigen}
    There exists a nonrandom continuously differentiable function $\delta$ defined on $\supp (F)=\cup_{k=1}^K [a_k, b_k]$ such that $\delta_n(x) \overset{a.s.}{\to} \delta (x)$ for $x\in \supp (F)$ as $n\to\infty$.
    In addition, the convergence is uniform for $x\in \cup_{k=1}^K [a_k+\eta, b_k-\eta]$ for any small $\eta>0$.
    Furthermore, there exists a finite nonrandom constant $M$  such that $|\delta_n(x)|$ is uniformly bounded 
away from zero by $M$ almost surely for all $x\in \cup_{k=1}^K [a_k-\eta, b_k+\eta]$, large $n$, and small $\eta>0$.
\end{assumption}


{
These assumptions are also adopted in \cite{ledoit2022quadratic} and \cite{ledoit2018optimal}.
Assumption \ref{assup_est_eigen} requires that shrinkage functions $\delta_n$ to be well behaved asymptotically.}
For a bounded, non-decreasing function $G$, we define its Steiltjes transform by $m_G(z) = \int_{-\infty}^\infty (x - z)^{-1} dG(x)$ for $z \in \mathbb{C}^+$, i.e. $z = x+iy$ with some $x \in \mathbb{R}$ and $y \in \mathbb{R}^+$. For a complex valued function $g(z)$, let $\mathrm{Re}[g(z)]$ and $\mathrm{Im}[g(z)]$ denote the real and imaginary parts, respectively. 

The following lemma is a key ingredient in finding the almost sure non-random limit of $\mathbf{L}_{m,n}$, where we obtain the limit of $\Phi_n^{(-l)}$ as $n\to\infty$.

\begin{lemma}\label{lem_converge_dis}
    Suppose Assumptions \ref{assup_ratio} (a) or (b), \ref{assup_sigma}, \ref{assup_iid} hold. 
    For any integer $l$, we have that $\Phi_n^{(-l)}$ converges almost surely pointwisely to $\Phi^{(-l)}(x)$ as $n\to\infty$ for all $x$ such that $\Phi^{(-l)}(x)$ is continuous, where
    \begin{align*}
    \Phi^{(-l)}(x) &= \lim_{\eta\to 0^+} \frac{1}{\pi} \int_{-\infty}^x \mathrm{Im} [\Theta^{(-l)} (\xi+i\eta)]d\xi\, , \, \text{and} \\
    \Theta^{(-l)}(z) &= \int_{-\infty}^{+\infty} \{\gamma[1-c^{-1}-c^{-1}z m_F(z)]-z\}^{-1}\gamma^{-l} dH(\gamma).
\end{align*}
\end{lemma}
\begin{proof}
For $x\in\mathbb{R}$,
\begin{align*}
    \Phi^{(-l)}_n(x) = \frac{1}{p} \sum_{i=1}^p \ind_{[\lambda_{n,i}, \infty)}(x) \sum_{j=1}^p \frac{(u_{n,i}^\T v_{n,j})^2}{\gamma_{n,j}^l}.
\end{align*}
By \cite[Lemma 6]{ledoit2011eigenvectors}, we have $\Phi^{(-l)}_n(x)\overset{a.s.}{\to}\Phi^{(-l)}(x)$ as $n\to\infty$. 

\end{proof}
{
Recall the Steiltjes transform $m_G(z)$ defined earlier. When $G$ has derivative $G'$, $m_G(z)$ has an extension to the real line $\breve{m}_G(x)  = \lim_{z \in \mathbb{C}^+ \to x} m_G(z)$ for any $x\in\mathbb{R}$ even though $\breve{m}_G(x)$ could be complex valued. We are now ready to state the almost sure non-random limit of $\mathbf{L}_{m,n}$ and the corresponding minimizer.
}
\begin{theorem}\label{thm_limit}
    Under Assumptions \ref{assup_ratio} (a), \ref{assup_sigma}, 
and \ref{assup_iid}, the loss $\mathbf{L}_{m,n}(\Sigma_n^{-1}, \tilde{\Sigma}_n^{-1}) 
    = \mathrm{tr}[(\Sigma_n^{-1} - \widetilde{\Sigma}_n^{-1})^2  S_n^{m}]/p$
     has the following almost sure limit:
    \begin{align}\label{limit}
    &\mathbf{L}_m =\int x^m d\Phi^{(-2)} (x) 
    - 2 \sum_{k=1}^K \int_{a_k}^{b_k} \frac{x^m}{\delta(x)} d\Phi^{(-1)}(x) 
    + \sum_{k=1}^K \int_{a_k}^{b_k} \frac{x^m}{\delta^2(x)} dF(x) \notag\\
    & = \int x^m d\Phi^{(-2)} (x) 
    - 2\sum_{k=1}^K \int_{a_k}^{b_k} \frac{x^m}{\delta(x)}\phi^{(-1)}(x) dF(x) 
    + \sum_{k=1}^K \int_{a_k}^{b_k} \frac{x^m}{\delta^2(x)} dF(x)
    \end{align}
    as $n\to\infty$,
    where 
    \begin{align*}
        \phi^{(-1)}(x) = 
          \begin{cases}
      0 & \text{if $x\le 0$}\\
      \dfrac{1 - c - 2cx \text{Re}[\breve{m}_F(x)]}{x} & \text{if $x>0$}.
    \end{cases} 
    \end{align*}
\end{theorem}
The proof is provided in the Appendix. 
{We also provide simulations for the convergence of the loss $\mathbf{L}_{1,n}$ in Section S.1 
of the supplement.
}

{
Since we focus on rotation invariant estimators $\widetilde{\Sigma}_n = U_n \widetilde{\Delta}_n U_n^\T$ with $\widetilde{\Delta}_n = \text{diag}(\delta_{n}(\lambda_{n,1}),$ $ \ldots, \delta_{n}(\lambda_{n,p}))$, minimizing the loss $\mathbf{L}_{m,n}(\Sigma_n^{-1}, \tilde{\Sigma}_n^{-1})$ with respect to the estimator $\widetilde{\Sigma}_n$ is equivalent to minimizing the loss 
with respect to the shrinkage function $\delta_n$.
Note that the limit of $\delta_n$ is $\delta$
by Assumption \ref{assup_est_eigen}, and that the limit of the loss $\mathbf{L}_{m,n}$ is $\mathbf{L}_m$ by Theorem \ref{thm_limit}.
Thus, minimizing the limiting loss $\mathbf{L}_m$ with respect to the shrinkage rule $\delta$ could lead to the best rotation invariant estimator for $\Sigma$ in terms of the limiting loss, which is provided in the following corollary.
}

\begin{corollary}\label{cor:oracle_estimator}
    The limit $\mathbf{L}_m$ in Theorem \ref{thm_limit} is minimized at 
    $$\delta^*(x) = \frac{1}{\phi^{(-1)}(x)} = \dfrac{x}{1 - c - 2cx \text{Re}[\breve{m}_F(x)]}, \quad \forall x \in \supp (F).$$ 
\end{corollary}
\noindent Therefore, the oracle 
best rotation invariant  estimator that is based on the true $F$ and minimizes the limiting loss is $\widetilde{\Sigma}_n^* = U_n \Delta^* U_n^\T$, where $\Delta^*$ is a diagonal matrix with diagonal elements $\delta^*$ as defined above. 
The denominator in the definition of $\delta^*$ is bounded away from zero under the same set of assumptions by \citet[Proposition 3.2]{ledoit2018optimal}.
Then it remains to construct a consistent estimator of $\delta^*$.
\begin{theorem}\label{thm:shrinkage_rule}
    Under the assumptions of Theorem \ref{thm_limit}, $\delta_n^*(x)$ defined as
    \begin{align}
        &\delta_{n}^* (x) = \left[ \left(1 - \frac{p}{n} \right) x^{-1} +  \frac{p}{n} x^{-1} 2 g^*_n(x^{-1}) \right]^{-1} 
        \text{ with } \nonumber \\  
        & g^*_n(x) = \frac{1}{p} \sum_{k=1}^p \lambda_{n,k}^{-1} \dfrac{\lambda_{n,k}^{-1} - x}{(\lambda_{n,k}^{-1} - x)^2 + h_n^2 \lambda_{n,k}^{-2}},
    \end{align}
   $h_n\sim C n^{-\alpha}$ for some $C>0$, and $\alpha\in(0,2/5)$ is a consistent estimator of $\delta^*(x)$, i.e. $\delta_n^*(x) \overset{p}{\to} \delta^*(x) $ for any $x \in \supp (F)$.
\end{theorem}
\noindent We provide the proof in the Appendix. It is instructive to note that the shrinkage rule $\delta_{n}^* (x)$ is exactly the same as the smoothed Stein shrinker in \cite{ledoit2022quadratic} where the authors show that this same rule is the optimum when one considers estimating $\Sigma_n$ under the Stein loss, i.e. $\mathbf{L}_n(\Sigma_n, \widetilde{\Sigma}_n) = p^{-1} \mathrm{tr}(\Sigma_n^{-1} \widetilde{\Sigma}_n) - p^{-1} \log |\Sigma_n^{-1}\widetilde{\Sigma}_n|$. 
{\begin{example}
    Consider the case when the data matrix $W_n \overset{iid}{\sim} \Gauss(0,1)$ and $\Sigma_n = \mathrm{I}$. For this case, the distribution of the population eigenvalues is a point mass at $1$, i.e. $H(x) = \mathbbm{1}(x \geq 1)$, and the limiting shrinkage function is the famous Marcenko-Pastur law: 
    $$dF(x) = \dfrac{1}{2\pi cx} \sqrt{(c_+ - x) (x - c_-)}\cdot \ind(c_- \leq x \leq c_+),$$
where $c_- = ( 1 - \sqrt{c})^2$ and $c_+ = (1 + \sqrt{c})^2$. 
It is well known that its Stieltjes transform $m_F(z) = (2cz)^{-1} (1 - c - z + \sqrt{(c_+ - z)(c_- - z)})$ for $z \in \mathbb{C^+}$, and hence 
$$\breve{m}_F(x) = \lim_{z \in \mathbb{C}^+ \to x} m_F(z) = \dfrac{1 - c - x + \sqrt{(c_+ - x)(c_- - x)}}{2cx}.$$
Thus, $\text{Re}[\breve{m}_F(x)] = (1 - c - x)/2cx$ for $x \in \text{Supp}(F)$, which implies that $\delta^*(x) = 1$. 
\end{example}
\begin{figure}
    \centering
    \includegraphics[width=0.9\textwidth, height = 6.5cm]{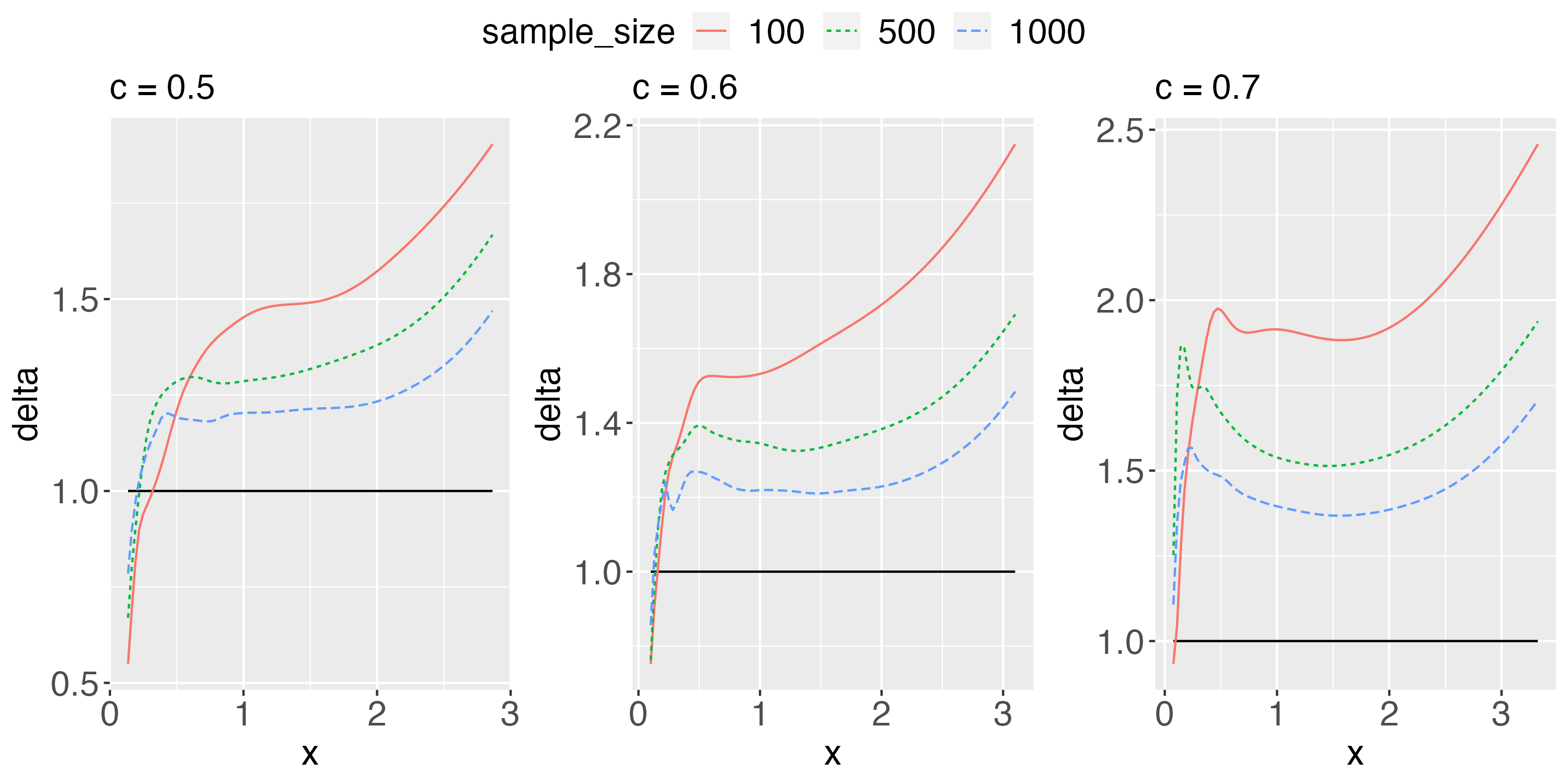}
    \caption{The population ($\delta^*(x)$) and sample shrinkage rule ($\delta(x)$) when the sample covariance matrix $S = n^{-1}XX^\T$ has entries $X_{ij} \overset{iid}{\sim} \Gauss(0,1)$. Three concentration ratios are considered $c = 0.5, 0.6, 0.7$. The black line corresponds to $\delta^*(x)$. The red, blue and chocolate lines correspond to $n = 100, 500, 1000$, respectively. }
    \label{fig:shrinkage_rules_main}
\end{figure}
In Figure \ref{fig:shrinkage_rules_main}, we plot $\delta^*(\cdot)$ and its estimator $\delta(x)$ (a realization based on $X_n$) for three choices of the concentration ratio $c = 0.5, 0.6, 0.7$ and $n = 100, 500, 1000$. {Here, the bandwidth parameter $h$ is chosen to obtain smooth decision rules (large $h$). This parameter plays a critical role in that smaller values lead to fluctuating decision rules. Additionally, for each panel, $h$ decreases with $n$ to ensure the condition of Theorem 6 is satisfied.} It can be seen from Figure \ref{fig:shrinkage_rules_main}, that the estimator $\delta^*$ converges to $\delta$ as $n$ grows. 
}

The shrinkage rule $\delta_n^*(x)$ has a bandwidth parameter $h_n$ which approaches 0 as $n \to \infty$ according to Theorem \ref{thm:shrinkage_rule}. \cite{ledoit2022quadratic} studied the dependence of the tuning parameter $h_n$ on the concentration ratio $c = p/n$ through  simulation experiments and suggested a default choice, which is $h_n = (c)^{0.7} (1/p)^{0.35}$. Although this choice works well in simulations, it is not data-driven.


In the following subsection, we provide a data-dependent choice of $h_n$ based on the estimation of the expected loss. While the results in this section till this point do not require any assumption about the distribution of $Z_n$ other than the existence of its moments, in the following subsection we assume $Z_n \sim \Gauss(0, \Sigma_n)$. 

\subsection{ (Nearly) 
Unbiased estimate of the risk}\label{sec:SURE}
{
In this subsection, we develop a nearly unbiased estimate of the expected loss when $m = 1$, and find the tuning parameter $h_n$ through minimizing the estimated expected loss.
}
Suppose Assumptions \ref{assup_ratio} (a), \ref{assup_sigma}, and \ref{assup_est_eigen} hold and in Assumption \ref{assup_iid} elements of $W_n$ are i.i.d $\mathrm{N}(0,1)$. Consider the estimator $\widetilde{\Sigma}_n = U_n \widetilde{\Delta}_n U_n^\T$ where $\widetilde{\Delta}_n$ is a diagonal matrix with diagonal elements $\delta_n^*(x)$ as defined in Theorem \ref{thm_limit}. The shrinkage rule $\delta_n^*(x)$ clearly depends on $h_n$. Hence, we write $\delta_n^*(x) = \delta_n^*(x;h_n)$ and the resulting estimator $\widetilde{\Sigma}_n = \widetilde{\Sigma}_n(h_n)$. 

Our goal here is to propose a nearly unbiased estimate of the risk $\mathbb{E}_\Sigma[\mathbf{L}_{m,n} (\Sigma_n^{-1}, \widetilde{\Sigma}_n^{-1}(h_n))]$ for any $h_n$ and finite $n$ when $m = 1$. Here, the expectation is computed with respect to the distribution of $Z_n$ which depends on the unknown $\Sigma_n$. Since here we are concerned with finite $n$, we suppress the dependence on $n$ and write $\Sigma_n = \Sigma$, $U_n=U$,
$\widetilde{\Delta}_n=\widetilde{\Delta}$,
$\delta_{n}^* (x)=\delta^* (x,h)$,
$\lambda_{n,j}=\lambda_j$,
$Z_n=Z$, $h_n=h$, $\widetilde{\Sigma}_n = \widetilde{\Sigma}(h)$ and $\mathbf{L}_{1,n} = \mathbf{L}(h) = \mathrm{tr}[\{\Sigma^{-1} - \widetilde{\Sigma}(h)^{-1}\}^2 S ]$ where $S = n^{-1}Z^\T Z$. The corresponding risk is $\mathbf{R}(\Sigma; h) = \mathbb{E}_{\Sigma} \{L(h)\} = n^{-1}\mathbb{E}_{\Sigma}\mathrm{tr}[\{\Sigma^{-1} - \widetilde{\Sigma}(h)^{-1}\}^2 (Z^\T Z) ]$. When $Z \sim \Gauss(0, \Sigma)$, we have $Z^\T Z \sim W_p(n, \Sigma)$ where $W_p(n, \Sigma)$ is the Wishart distribution of dimension $p$, degrees of freedom $n$ and parameter $\Sigma$. 

Let $\underline{S} = Z^\T Z$. Consider the spectral decomposition of $S = U \Lambda U^\T$ which implies that the spectral decomposition of $\underline{S} = n S = U (n \Lambda) U^\T = U \Lambda^* U^\T$ with $\lambda_j^* = n\lambda_j$. Since $\widetilde{\Sigma}(h) = U \widetilde{\Delta}(h) U^\T $ with $\widetilde{\Delta}(h)=\text{diag}(\delta^*(\lambda_{1}, h), \ldots, \delta^*(\lambda_{p}, h))$, 
we have
\begin{align}\label{equ_risk}
    &\mathbf{R}(\Sigma; h) = n^{-1}\mathbb{E}_{\Sigma}\left[ \mathrm{tr}\left( \widetilde{\Sigma}(h)^{-2} \underline{S} \right) - 2 \mathrm{tr}\left( \Sigma^{-1}  \underline{S} \widetilde{\Sigma}(h) ^{-1} \right) + \mathrm{tr}\left( \Sigma ^{-2}  \underline{S}\right)\right] \notag\\
    & = n^{-1} \sum_{j=1}^p \mathbb{E}_{\Sigma}\left[ \dfrac{\lambda_j^*}{\delta^*(\lambda_j, h)^2} \right] - 2 n^{-1} \mathbb{E}_{\Sigma}\left[\mathrm{tr}\left(\Sigma^{-1}  \underline{S} \widetilde{\Sigma}(h) ^{-1} \right)\right] \nonumber \\  + & n^{-1}\mathbb{E}_{\Sigma} \left[\mathrm{tr}\left( \Sigma ^{-2}  \underline{S}\right)\right] .
\end{align}
In the above display, the first term inside the expectation is computable for any observed $S$ and given $h$, and it is trivially an unbiased estimate of its own expectation. We next focus on  the third term $\mathbb{E}_\Sigma \left[ \mathrm{tr}(\Sigma^{-2} \underline{S})\right] = \mathbb{E}_\Sigma \left[ \mathrm{tr}(\Sigma^{-1} \underline{S} \Sigma^{-1})\right] = \mathbb{E}_\Sigma \left[ \mathrm{tr}(\underline{S}^*)\right]$ where $\underline{S}^*=\Sigma^{-1} \underline{S} \Sigma^{-1} \sim \mathrm{W}_p(n, \Sigma^{-1})$. Let $\underline{\Sigma}=\Sigma^{-1}$. Then from properties of the Wishart distribution it is immediate that $\mathbb{E}_\Sigma \left[ \mathrm{tr}(\underline{S}^*)\right] = n \sum_{j=1}^p \underline{\Sigma}_{jj}$, where $\underline{\Sigma}_{jj}$ is the $j$-th diagonal element in $\underline{\Sigma}$. 
Since in the current setup $Z_i \sim \Gauss(0, \Sigma), \, i=1, \ldots, n$, we have that $Z_{ij}\mid Z_{i,-j} \sim \Gauss(\alpha_j^\T Z_{i,-j}, \underline{\Sigma}_{jj}^{-1})$ where $Z_{i, -j}$ is the vector obtained by excluding the $j$-th element from $Z_i$. Thus, an unbiased estimator of $\underline{\Sigma}_{jj}^{-1}$ is available from the error variance estimate obtained by linearly regressing the $j$-th column of $Z$ on the rest of its columns. Inverting this estimator gives us a reasonably good estimator of $\underline{\Sigma}_{jj}$, denoted by $\widehat{\underline{\Sigma}}_{jj}$, although strict unbiasedness can be maintained by employing the so-called sum-estimator \citep{glynn2014exact}. {However, this sum-estimator might result in increased computational burden, and in our numerical experiments we observed the biased estimator to perform better than the unbiased estimator} .

The key issue of estimating the risk is the second term in Equation (\ref{equ_risk}), which is dealt with in the following theorem. {The result is an application of Haff's identity \citep{haff1979estimation} suited to our context.} The proof is given in the Appendix.
\begin{theorem}\label{thm:SURE}
    Under Assumptions \ref{assup_ratio} (a), \ref{assup_sigma}, \ref{assup_iid}, and when elements of $W_n$ are distributed $\Gauss(0,1)$ independently, we have
    $$\mathbb{E}_{\Sigma} \left[\mathrm{tr}\left(\Sigma^{-1}  \underline{S} \widetilde{\Sigma}(h) ^{-1} \right)\right] =  \mathbb{E}_{\Sigma} \left[ \mathrm{tr}\left\lbrace (n-p - 1) \widetilde{\Sigma}(h) ^{-1} + 2 D_{\underline{S}}([\widetilde{\Sigma}(h)^{-1}]^{\T} \underline{S})\right\rbrace\right], $$
    where $D_{\underline{S}}$ is a differential operator defined as $D_{\underline{S}} = \left\lbrace \frac{1}{2}(1 + d_{ij}) \frac{\partial}{\partial \underline{S }_{ij}}\right\rbrace$ for $1\leq i,j\leq p$ with $d_{ij} = 1$ if $i =j$ and $0$ otherwise.
\end{theorem}

Therefore, we construct the following nearly unbiased estimate of the risk of the estimator $\widetilde{\Sigma}(h)$: 
\begin{align*}
    \widehat{R}(h) =& n^{-1} \sum_{j=1}^p \left[ \dfrac{\lambda_j^*}{\delta^*(\lambda_j,h)^2} \right] - 2n^{-1} \left[ \mathrm{tr}\left\lbrace (n-p - 1) \widetilde{\Sigma}(h) ^{-1} \right\rbrace\right] \\ 
     & - 4n^{-1} \left[\mathrm{tr} \left\lbrace D_{\underline{S}}([\widetilde{\Sigma}(h)^{-1}]^{\T} \underline{S})\right\rbrace\right] + \sum_{j=1}^p \widehat{\underline{\Sigma}}_{jj}.
\end{align*}
The proposed value $\hat{h}$ for the tuning parameter $h$ is the minimizer of  $\widehat{R}(h)$.
\subsection{Shrinkage in the $p>n$ case}\label{sec:p>n} 
In this subsection, we extend our results for the relative savings loss to the $p>n$ case, where the sample covariance matrix is singular. The overall strategy remains the same. That is, we find an almost sure limit of the loss, find its minimizer, and then construct a consistent estimator of the minimizer. In this case, the limiting spectral distribution of the sample eigenvalues $F$ is a mixture of a point mass at 0 and a continuous component, i.e. $F(x) = \{(c-1)/c\} \mathbb{I}_{[0,\infty]}(x) + (1/c)\underline{F}(x)$ where $\underline{F}(x)$ is the continuous component whose support is bounded away from 0. The Steiltjes transform of $F$ and $\underline{F}$ are related as $\breve{m}_{\underline{F}}(x) = c\breve{m}_{F}(x) + (c-1)/x$ for all $x \in \mathbb{R}$.
From \cite[Lemma 14.1]{ledoit2018optimal} we have $\Phi_n^{(-1)}(x)$ converges almost surely to $\Phi^{(-1)(x)}$ for all $x \in \mathbb{R} - \{0\}$. And, $\Phi^{(-1)}$ is continuously differentiable on $\mathbb{R} - \{0\}$ and can be expressed for all $x \in \mathbb{R}$ as
$\Phi^{-(1)}(x) = \int_{-\infty}^x \phi^{(-1)}(u)dF(u)$ where
\begin{align*}
        \phi^{(-1)}(x) = 
          \begin{cases}
      0 & \text{if $x < 0$}\\
      \frac{c}{c - 1} \breve{m}_H(0) - \breve{m}_{\underline{F}}(0) & \text{if $x = 0$} \\
      \dfrac{1 - c - 2cx \text{Re}[\breve{m}_F(x)]}{x} & \text{if $x >0$}.
    \end{cases} 
    \end{align*}
We then have the following limit of the relative savings loss, noting that in this case $dF(x) = (1/c)d\underline{F}(x)$ for $x>0$.
\begin{theorem}\label{thm:p>n_limit}
  Under Assumptions \ref{assup_ratio} (b), \ref{assup_sigma}, \ref{assup_iid} and in the case when $c>1$, $\mathbf{L}_{m, n} = \mathrm{tr}[(\Sigma_n^{-1} - \widetilde{\Sigma}_n^{-1})^2  S_n^{m}]/p$ has the following almost sure limit:
  \begin{align*}
  \mathbf{L}_m = \begin{cases}
  &\int_{-\infty}^{\infty} x^m d\Phi^{(-2)} (x) - 2\sum_{k=1}^K \int_{a_k}^{b_k} \frac{x^m}{\delta(x)} \phi^{(-1)}(x) dF(x) \\
   & + \sum_{k=1}^K \int_{a_k}^{b_k} \frac{x^m}{\delta^2(x)} dF(x), \,\, m\geq 1\\
  &  \int_{-\infty}^\infty d\Phi^{(-2)}(x) - 2 \sum_{k=1}^K \int_{a_k}^{b_k} \frac{1}{\delta(x)} \phi^{(-1)}(x) d F(x)  \\
  &+ \sum_{k=1}^K \int_{a_k}^{b_k} \frac{1}{\delta^2(x)}dF(x) + \frac{c-1}{c} \left[ \frac{1}{\delta^2(0)} - 2 \frac{\frac{c}{c-1}\breve{m}_H(0) - \breve{m}_{\underline{F}(0)}}{\delta(0)}\right], \,\, m = 0.
  \end{cases}
  \end{align*} 
\end{theorem}
The above limit immediately yields the following oracle decision rule, specifically for the case $m\geq1$.
\begin{corollary}\label{cor:oracle_estimator_p>n}
    For $m\geq1$, the limit $\mathbf{L}_{m}$ in Theorem \ref{thm:p>n_limit} is minimized at $\delta^*(x) = \frac{1}{\phi^{(-1)}(x)}$ for any $x>0$. 
\end{corollary}
The next result gives us a consistent estimator of this optimal shrinkage rule.
\begin{theorem}\label{thm:shrinkage_rule_p>n}
    Under the assumptions of Theorem \ref{thm:p>n_limit}, $\delta_n^*(x)$ for $x>0$, defined as
    \begin{align*}
        &\delta_{n}^* (x) = \left[ \left(\frac{p}{n} - 1 \right) x^{-1} +   2x^{-1}  g^*_n(x^{-1}) \right]^{-1}, \\
        & g^*_n(x) = \frac{1}{n} \sum_{k=p-n+1}^p \lambda_{n,k}^{-1} \dfrac{\lambda_{n,k}^{-1} - x}{(\lambda_{n,k}^{-1} - x)^2 + h_n^2 \lambda_{n,k}^{-2}},
    \end{align*}
   with $h_n\sim C n^{-\alpha}$ for some $C>0$ and $\alpha\in(0,2/5)$ is a consistent estimator of $\delta^*(x)$, i.e. $\delta_n^*(x) \overset{p}{\to} \delta^*(x) $ for any $x \in \supp (\underline{F})$.
\end{theorem}

{
The above two results imply that, when $m\ge 1$, the minimizer does not rely on the spectral distribution $H$ and the shrinkage rule is only defined for the non-zero eigenvalues of the sample covariance matrix. 
This is due to the special nature of the relative savings loss for the case $m\ge1$ as shown in Theorem \ref{thm:p>n_limit}.
Indeed, when $p>n$, the sample covariance matrix $S_n$ has $n$ non-zero eigenvalues and $p-n$ zero eigenvalues. This implies for $m\geq 1$, and any $\widetilde{\Sigma}_n = U_n \widetilde{\Delta}_n U_n^\T$,
we have
$$\mathbf{L}_{m,n}(\Sigma_n^{-1}, \widetilde{\Sigma}_n^{-1}) = \frac{1}{p} \sum_{j=1}^n (u_{n,j}^\T \Sigma_n^{-2} u_{n,j}) \lambda_{n,j}^m - \frac{2}{p}\sum_{j=1}^n (u_{n,j}^\T \Sigma_n^{-1}u_{n,j}) \frac{\lambda_{n,j}^m}{\delta_{n}(\lambda_{n,j})} + \frac{1}{p} \sum_{j=1}^n \frac{\lambda_{n,j}^m}{\delta_{n}^2(\lambda_{n,j})}.$$
In other words, the zero sample eigenvalues do not contribute to the loss.
When $m=0$,
the oracle decision rule 
may involve 
the population eigenvalues according to Theorem \ref{thm:p>n_limit} through the term $\breve{m}_H(0)$, which can be estimated by
inversion of the QuEst function \citep{ledoit2018optimal}.
In fact,  \cite{ledoit2022quadratic} and \cite{ledoit2018optimal} have considered this covariance estimation problem for $p>n$ under other loss functions, among which some loss functions such as Stein's loss also needs numerical inversion of the QuEst function.
}

For our purpose of estimating the parameter matrix $B$ in Section \ref{sec:estimator}, we are specifically interested in the case $m=1$. 
Although decision rules for null eigenvalues do not affect the loss, for our primary objective of estimating $B$ by $\widehat{B}(1 - \widehat{\Sigma}^{-1}Q)$, we do need $\widehat{\Sigma}^{-1}$, which would not be possible to compute with zero eigenvalues. For this reason, we consider the following shrinkage rule for zero eigenvalues:
\begin{equation}
    \delta_n^*(x)^{-1} = \left(\frac{p}{n}-1\right) \frac{1}{n} \sum_{j= p-n+1}^p \lambda_{n,j}^{-1} \quad \text{for } x = 0,
\end{equation}
which is the optimal shrinkage rule for estimation under Frobenius loss \cite[Section 5]{ledoit2022quadratic}.

\section{Coefficient matrix $B$ estimation }\label{sec:local_linear_shrinkage}
We now again go back to our original problem of estimating the coefficient matrix $B$. According to the results in Section \ref{sec:covariance_shrinkage}, we propose 
to use 
\begin{align}\label{Equ_proposed_estimator}
\bar{B}=\widehat{B}(\mathrm{I} -   Q^{1/2} \overline{\Sigma}_n Q^{1/2})
\end{align}
with 
$\overline{\Sigma}_n = U_n \text{diag}(\delta_{n}^*(\lambda_{n,1}),$ $ \ldots, \delta_{n}^*(\lambda_{n,p})) U_n^\T$ to estimate $B$, where $\delta_{n}^*(x)$ is defined in Theorem \ref{thm:shrinkage_rule} {with $h_n=\hat{h}$ as defined in Section \ref{sec:p>n}} if $p<n$ and is defined in Theorem \ref{thm:shrinkage_rule_p>n} if $p>n$.
This estimator represents a global shrinkage rule on the data matrix $\widehat{B}$.

In this section, we extend this global shrinkage estimator of $\beta^{(t)}$ to a local one with  a mixture prior distribution. {
Scale mixtures of Gaussian \citep{polson2010shrink} are the cornerstone of modern Bayesian sparse modeling techniques. The two component spike-slab prior is also a mixture of Gaussian random variables with appropriately small variance \citep{narisetty2014bayesian}.
These priors are obtained by mixing a Gaussian (typically with mean 0) with a scale parameter $\lambda$ and some mixing distribution $g$ over $\lambda$. Careful choices of $g$ lead to flexible prior distributions with sufficiently heavier tails than a Normal distribution. Posterior means from these priors apply different levels of shrinkage depending on the magnitude of the signal. In other words, the shrinkage function is not global.}

In particular, we consider the prior $\beta^{(t)} \sim \sum_{k=1}^K \pi_k \Gauss(0, Q^{1/2}\Omega_k Q^{1/2})$, so that the posterior mean is a weighted combination of linear shrinkage rules. Indeed, when $\beta^{(t)} \sim \sum_{k=1}^K \pi_k \Gauss(0, Q^{1/2}\Omega_k Q^{1/2})$, we have
\begin{align*}
    \mathbb{E}(\beta^{(t)}\mid \widehat{B}) &= (\mathrm{I} - \sum_{k=1}^K \pi_k^*C_k)\widehat{\beta}^{(t)}, \quad C_k = Q^{1/2} \Sigma_k^{-1} Q^{-1/2}, \Sigma_k = \mathrm{I} + \Omega_k \\
    \pi_k^\star & = \dfrac{\pi_k f(\widehat{\beta}^{(t)}; 0, Q^{1/2}\Sigma_k Q^{1/2})}{ \sum_{k=1}^K \pi_k f(\widehat{\beta}^{(t)}); 0, Q^{1/2}\Sigma_k Q^{1/2})},
\end{align*}
where $f(x; \mu, \Sigma)$ is the multivariate Gaussian density evaluated at $x$ with mean $\mu $ and covariance $\Sigma$. Moreover, the marginal distribution of the observed data is a mixture of Gaussian as well: $\widehat{\beta}^{(t)} \sim \sum_{k=1}^K \pi_k^* \Gauss(0, Q^{1/2}\Sigma_kQ^{1/2})$. 
The posterior distribution of $\beta^{(t)}$ is also available in closed form: $\beta^{(t)}\mid \widehat{\beta}^{(t)} \sim \sum_{k=1}\pi_k^\star \Gauss\{(\mathrm{I} - C_k)\widehat{\beta}^{(t)}, (\mathrm{I} - C_k)Q\}$. 
Note that more weight is given to component-specific posterior means that are supported by the data through the posterior probabilities. 

As in Section \ref{sec:estimator}, we consider the case when the prior parameters $\Omega_k$ are unknown.
In order to obtain the posterior mean $\mathbb{E}(\beta^{(t)}\mid \widehat{B})$, we need estimates of $\Sigma_k$. Recall that $\widehat{\beta}_\star^{(t)} = Q^{-1/2} \widehat{\beta}^{(t)}$. To estimate $\Sigma_k$, we introduce latent indicators $Z \in \{0,1\}^{K \times 1}$ such that $\widehat{\beta}_\star^{(t)}\mid Z_k = 1 \sim \Gauss(0, \Sigma_k)$. Thus, given these indicators, estimating $\Sigma_k$ can be achieved by resorting to covariance matrix estimates  obtained in Section \ref{sec:covariance_shrinkage}. Since the indicators are unknown, we use a Markov chain Monte Carlo (MCMC) sampler to sample $Z \mid \Sigma_k$, and then estimate $\Sigma_k$ conditional on the sampled indicators $Z$ - the standard data augmentation technique to fit Bayesian mixture models \citep{hobert2011data}. The resulting algorithm is summarized in Algorithm S.9 
of the supplement. 
{Regarding the choice of $K$ in practice, we suggest dividing the data into training and validation subsets, and determining the optimal $K$ by comparing the prediction errors of the proposed method on the validation set for various candidate values of $K$.
}

{
In addition, we provide a discussion of the risk of the proposed global and local srinkage rules in Section S.2  
of the supplement, which quantifies the potential risk reduction for the proposed estimators compared to the maximum likelihood estimator.
}

\section{Simulation experiments}\label{sec:experiments}
In this section, we conduct a series of extensive simulation experiments, and compare the proposed method with existing methods in the literature. 
We focus on the performance of the linear shrinkage estimator from Section \ref{sec:local_linear_shrinkage} in estimating the parameter matrix $B$. \texttt{R} codes to implement our method are available at \href{https://github.com/antik015/EBayes-Integration}{https://github.com/antik015/EBayes-Integration}.

To benchmark the performance of the proposed method, we compare it to the following list of methods - 1) Unified Test for MOlecular SignaTures (UTMOST) developed by \cite{hu2019statistical} for cross-tissue TWAS models where the estimator is obtained by optimizing an error sum of squares for all tissues with added $\ell_1$-penalties on the columns of the mean matrix, and $\ell_2$-penalties on the rows of the mean matrix, 2) Iterated stable autoencoder (ISA) developed by \cite{josse2016bootstrap} where the estimator is obtained by creating an autoencoder of the covariate matrix using a careful bootstrap scheme, and 3) Multivariate Adaptive Shrinkage (MASH) developed by \cite{kim2024flexible} where continuous shrinkage priors are approximated by finite mixtures of Gaussian distributions, and variational methods for scalable computation are combined to estimate the mean matrix. The first two methods make a working structural assumption on the parameter to be estimated, in that it is either sparse or low-rank. The last method is under the empirical Bayes framework.

We consider $N = 200$ points within each source of data (i.e., tissue in the context of TWAS), and $n = 40$ or $50$ as the number of data sources. Within each source, we vary the number of predictors to be $p = 10, 20,$ or $30$ when $n = 40$, and $p = 20, 30,$ or $40$ when $n = 50$. {We provide additional simulations for high-dimensional settings in Section S.4 
of the supplement. }

Each row of the matrix of covariates $X$ for all combinations was generated independently from $\Gauss(0, \Sigma_X)$ where $\Sigma_X  = (1-\rho) \mathrm{I}_p + \rho \mathbf{J}_p$ where $\mathbf{J}_p$ is a $p \times p$ matrix with all entries equal to 1. We consider $\rho = 0, 0.5,$ or $0.8$ in our experiments. The true parameter matrix $B_0\in \mathbb{R}^{p \times n}$ with columns $\beta_0^{(t)}, \, t = 1, \ldots, n$ were generated using four specific settings - a) low-rank (LR) where $B_0 = F_0 G_0^\T$, and $F_0 \in \mathbb{R}^{p \times r}$, $G_0 \in \mathbb{R}^{n \times r}$ with $r = 8$, b) approximate sparse (AS) where each entry of $B_0$ is generated independently from $\Gauss(0, \tau_0^2)$ with $\tau_0 = 0.2$, so that $\mathbb{E}[\norm{B_0}_F^2] =0.2^2 np$, c) Horseshoe (HS) where elements of $\beta_0^{(t)}$ are generated from the Horseshoe prior \citep{carvalho2010horseshoe} by first generating $\tau \sim \text{Unif}[0,1]$, then $\lambda_j \sim C^{+}(0,1)$, and $\beta_{0j}^{(t)}\mid \tau, \lambda_j \sim \Gauss(0, \lambda_j^2\tau^2)$ where $C^+(0,1)$ is the standard Half-Cauchy distribution with pdf $f(x) \propto (1+ x^2)^{-1}$, d) mixture (Mix) where $\beta^{(t)}_0 \sim \pi_1 \Gauss(0, A_1) + \pi_2 \Gauss(0, A_2) $ with $A_1 = 0.1^2 \mathrm{I}_p$, and $A_2 = 10^2 \mathrm{I}_p$. Each of the design settings considered here focuses on specific cases. For example, the low-rank setting is favorable to methods that are designed to recover parameters with this structure, the Horseshoe setting is designed to have a mix of signal strengths within each column of the true coefficient matrix. On the other hand, in the approximate sparse setting all elements of $B_0$ are very small and methods that penalize the size of the coefficient matrix are expected to perform well, such as UTMOST. Finally, the mixture setting generates coefficients that have strong signals for some data sources and weak signals for remaining data sources, a typical scenario in TWAS.  

We then generate  the response within each source as $y^{(t)} = X \beta_0^{(t)} + \epsilon^{(t)}$ where $\epsilon^{(t)} \sim \Gauss(0, 1)$. 
For every method, we compute the mean square error (MSE) computed as $(pn)^{-1}\norm{\widetilde{B} - B_0}_F^2$ and predictive mean squared error (PE) computed as $(nN_t)^{-1}\norm{X_t\widetilde{B} - X_t B_0}_F^2$, averaged over $20$ independent replications within each setting, where $\widetilde{B}$ is an estimator for $B_0$. Here $X_t$ is a $N_t \times p$ matrix of testing data on the covariates; we set $N_t =20$. 
For the proposed local linear shrinkage estimator, we consider three choices of $K = 2, 3, 4$ abbreviated as LLS-2, LLS-3, LLS-4.  
The proposed global linear shrinkage estimator defined in Equation \eqref{Equ_proposed_estimator} is abbreviated as ULS for unbiased linear shrinkage. We also consider the proposed global linear shrinkage estimator with the default smoothing parameter suggested by \cite{ledoit2022quadratic} as the $h_n$, which is abbreviated as LS.


\begin{figure}
    \centering
    \includegraphics[height = 7.5cm, width = 0.85\textwidth]{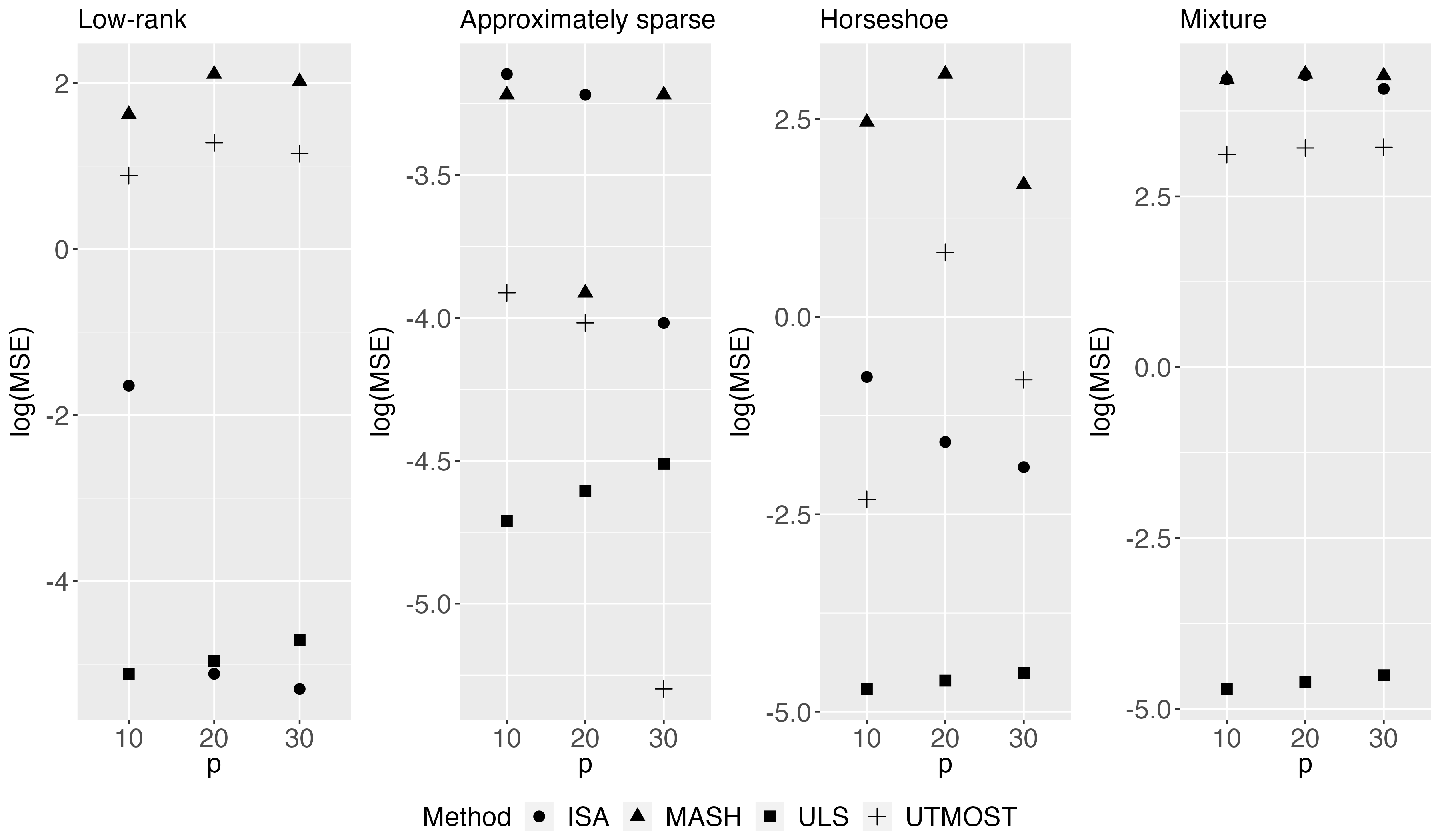}
    \caption{Logarithm of MSE when $n = 40$ and $\rho = 0.5$ under the four experimental settings.}
    \label{fig:MSE_T=40}
\end{figure}
\begin{figure}
    \centering
    \includegraphics[height = 7.5cm, width = 0.85\textwidth]{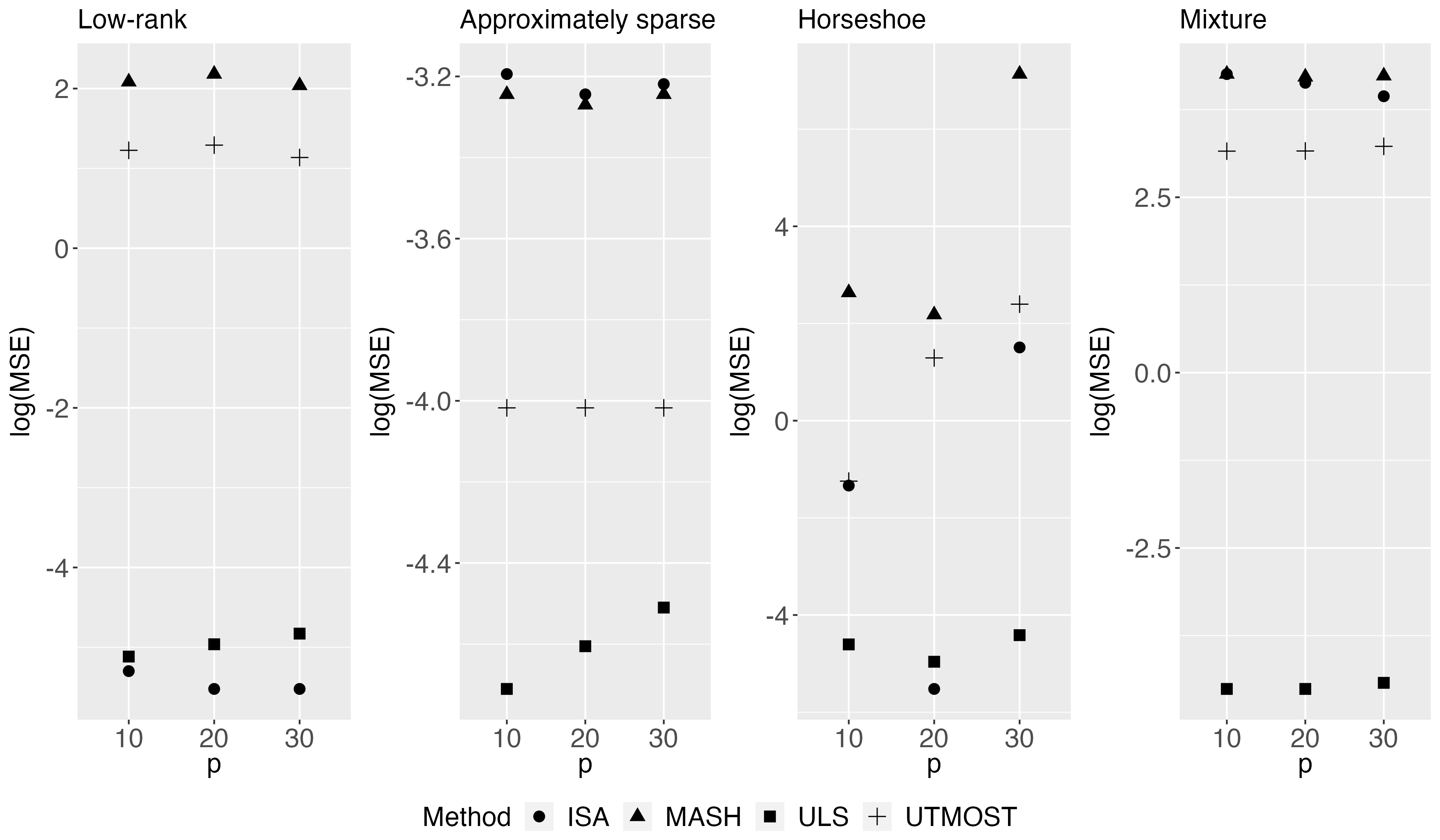}
    \caption{Logarithm of MSE when $n = 50$ and $\rho = 0.5$ under the four experimental settings.}
    \label{fig:MSE_T=50}
\end{figure}
\begin{figure}
    \centering
    \includegraphics[height = 7.5cm, width = 0.85\textwidth]{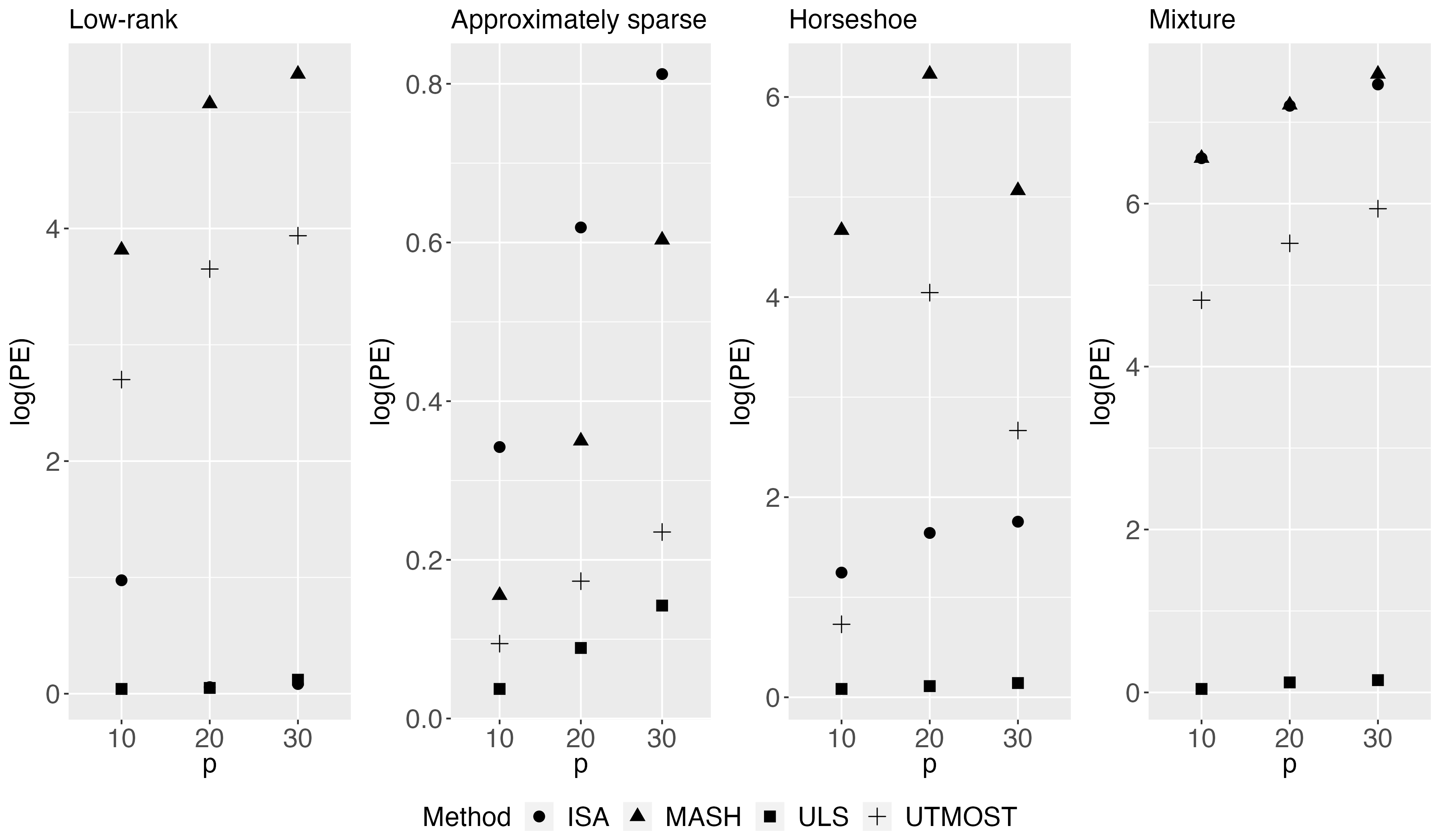}
    \caption{Logarithm of PE when $n = 40$ and $\rho = 0.5$ under the four experimental settings.}
    \label{fig:PE_T=40}
\end{figure}
\begin{figure}
    \centering
    \includegraphics[height = 7.5cm, width = 0.85\textwidth]{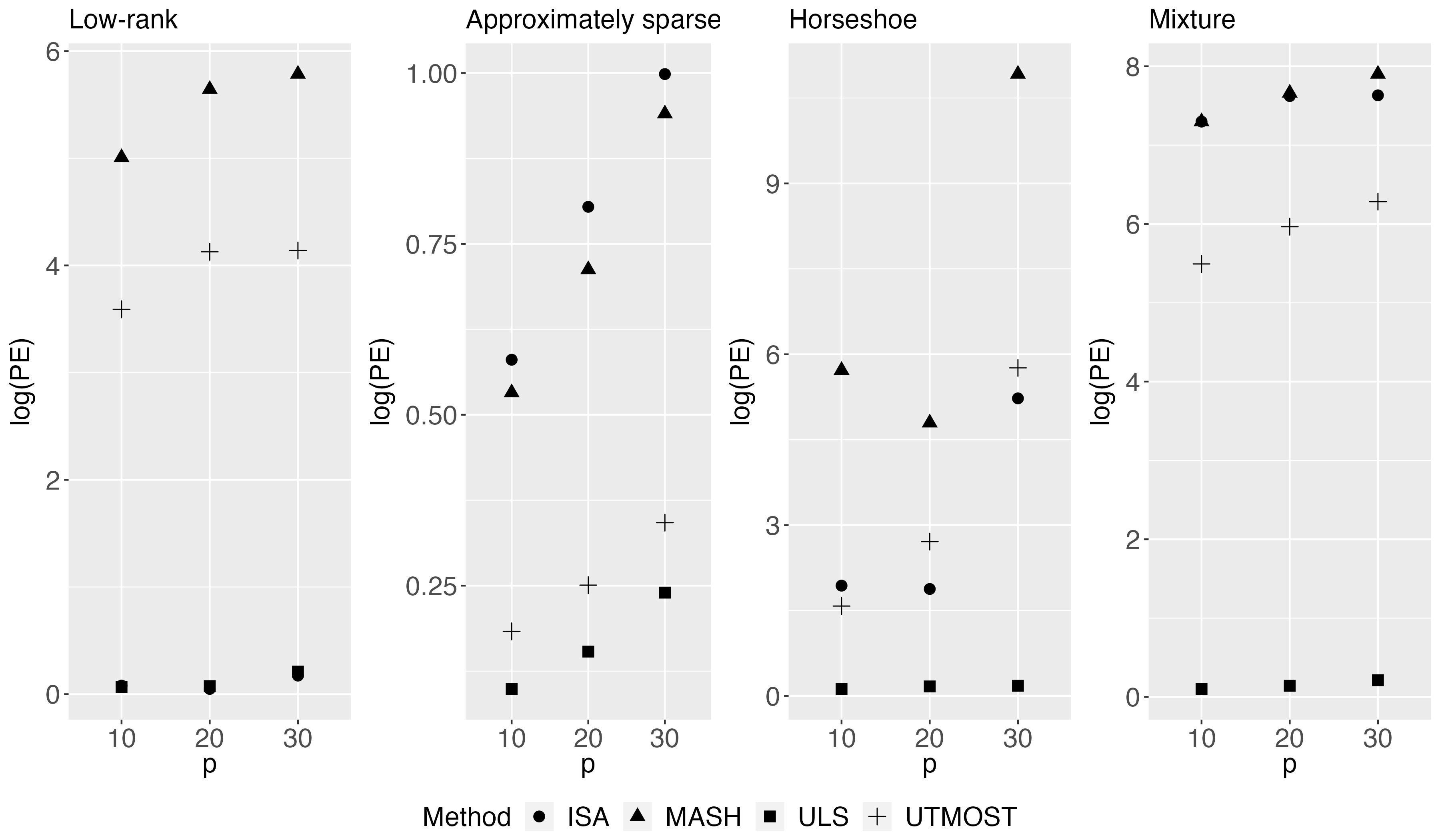}
    \caption{Logarithm of PE when $n = 50$ and $\rho = 0.5$ under the four experimental settings.}
    \label{fig:PE_T=50}
\end{figure}

The detailed results for all the methods are reported in Tables 
S.6--S.9,
which can be found in the Appendix.
In particular, we compare the proposed ULS method with existing methods in Figures \ref{fig:MSE_T=40}--\ref{fig:PE_T=50}.
Naturally, when the true parameter matrix has a low-rank structure, the best performing method is ISA, although errors for ULS are very close as shown in the figures. In other cases, ULS is the best performing method, especially when it comes to MSE. 
As shown in Tables S.6--S.9,
other proposed estimators also maintains good performance across different settings. 
When the true parameter is generated from a two-component mixture, errors of LLS with $K = 3, 4$ are very close to errors of $K =2$.
{In addition, we expect the level of error of the proposed method to increase as $p$ grows, since Figure S.1 
in the supplement shows that the value of the loss function $\mathbf{L}_{1,n}$ increases as $c=\lim_{n\to\infty} p/n$ increases for a fixed $n$.
}

{
We observe a lot of variation of errors for the Horseshoe case. The Horseshoe prior is a heavy-tailed prior (polynomially decreasing tails) with an infinite spike at 0. As a result, samples from this distribution can be very small and very large. The MASH/ISA/UTMOST algorithm is particularly sensitive to these kind of parameter settings and hence the large variations. In particular, the variation comes from the fact that in some of the replications, these algorithms did not converge.
}


 Overall, the results clearly point towards the versatility of the proposed shrinkage method, and highlights the benefits of the approach in that when the underlying structure of the parameter is unknown, it is perhaps better to shrink. This approach may not always provide the best performance. However, one can still achieve reasonable performance in a variety of settings.

\section{Real data application}\label{sec:real_data}

In this section, we apply the proposed method to the GTEx data and compare it with the existing OLS,  UTMOST, ISA, and MASH methods in  predicting tissue-specific expressions using cis-SNPs and estimating the corresponding effects.

We aim to study the relationship between 
 the cis-SNPs and gene expression across various tissues and genes. 
  The GTEx project, initiated in 2010 as part of the NIH Common Fund, provides a comprehensive public resource database to study the association between genetic variation and gene expression.
  This project has collected genotype and gene expression data from $838$ participants in $49$ tissue types, which are extracted from tissue samples by the Laboratory, Data Analysis and Coordinating Center (LDACC) \citep{lonsdale2013genotype}.
  However, we do not have tissue samples of all the $49$ tissue types from each participant.
  In this paper, we focus only on 32 tissues, each of which contains at least $200$ samples.

  We obtain gene expression values from ``GTEx\_Analysis\_v8\_eQTL\_expression\_matrices
  tar'' at \url{https://www.gtexportal.org/home/downloads/adult-gtex/qtl}.  These values have been fully processed and normalized by the GTEx project \citep{gtex2020gtex, lonsdale2013genotype}. The genotype data in ``GTEx\_Analysis\_2017-06-05\_v8\_Who\_leGenomeSeq\_838 Indiv\_Analysis\_Freeze.SHAPEIT2phased.vcf.g'' are also used and 
  processsed 
  as outlined below. The SNPs with minor allele frequencies less than $5\%$ are excluded. The remaining SNPs are filtered for linkage disequilibrium (LD) using PLINK 1.9 
  with a window size of 50 SNPs, a step size of 5 SNPs, and an $R^2$ threshold of 0.2.
  Furthermore, we identify cis-eQTLs for each gene based on the approach in \cite{wang2016imputing}.
  As a result, 3,201 genes have {at least two and at most twenty} associated cis-SNPs.

We adopt the proposed LS method, along with the OLS, UTMOST, ISA, and MASH methods, to predict gene expression based on the corresponding cis-SNPs for each gene. 
{
Here we only use the LS method since other proposed methods perform similarly as the LS based on the simulations in Section \ref{sec:experiments}.}
To assess the prediction performance of each method,
we conduct a 10-fold cross-validation analysis for each gene. Specifically, for each tissue, we randomly split all the observed samples into 10 equally sized folds, and name them Folds 1--10. For each $i=1,\dots, 10$, we treat Fold $i$ in all the tissues as a testing set, and the remaining folds in all the tissues together as a training set.
We then use the average prediction mean squared error (PMSE) across all the folds and genes for the evaluation of prediction accuracy.

We provide the average PMSE across all the 3,201 genes for each method in the first row of Table \ref{Table_PMSE}, which shows that the proposed method, UTMOST, and MASH perform slightly better than the OLS method.
In the second and third rows of Table \ref{Table_PMSE}, we  provide the average PMSEs across 2,736 and 2,019 genes which have at most 10 and 5 cis-SNPs, respectively, showing that the proposed method and the UTMOST perform the best among all the methods. The MASH performs worse than these two methods when there are fewer cis-SNPs.

To further compare all the methods, we plot heat maps of estimated cis-SNP effects by all the methods for the ZNF138 gene, which is shown in Figure \ref{heatmap}. In sub-figures of Figure \ref{heatmap}, each row of the heat map represents a tissue and each column represents a cis-SNP.
We can observe that the MASH and UTMOST methods tend to shrink SNP effects to zero, while the proposed method allows non-zero small effects. Moreover, the MASH method is likely to fuse effects from different tissues to the same value, while the proposed method enable variability of effects across tissues.
In fact, there could exist considerable differences among the cis-SNP effects in different tissue types \citep{fu2012unraveling}.
Furthermore, gene expression could be complex and controlled by many SNPs with small effects \citep{heap2009complex, gresle2020multiple, lloyd2017genetic, boyle2017expanded}.

{
In addition, we have applied the
proposed LS method to the Yeast Cell Cycle dataset and compared that with existing methods. The results are provided in Section S.5 
of the supplement. The proposed method works well in the multi-response regression for the
Yeast Cell Cycle dataset.
}

\begin{table}[hbt!]
\begin{center}
\caption{Average prediction mean squared errors (PMSEs) for different methods, with standard deviation (SD) in the parentheses. 
``PMSE'' represents the average PMSE across all the 3,201 genes.
``$\text{PMSE}_{10}$'' and ``$\text{PMSE}_5$'' 
represent average PMSEs for 
genes with at most $10$ and $5$ cis-SNPs.
}\label{Table_PMSE}
\vspace{2mm}
\resizebox{\columnwidth}{!}{%
\begin{tabular}{l|ccccc}
\hline
Methods & Proposed (LS) & OLS & UTMOST & ISA & MASH \\ 
\hline
PMSE & 0.964 (0.018) & 0.968 (0.017)  & 0.962 (0.016) & 0.982 (0.023) & 0.964 (0.024) \\
\hline
$\text{PMSE}_{10}$ & 0.966 (0.011) & 0.970 (0.011)  & 0.966 (0.010) & 0.987 (0.015) & 0.970 (0.015) \\
\hline
$\text{PMSE}_5$ & 0.969 (0.007) & 0.972 (0.007)  & 0.969 (0.006) & 0.988 (0.010) & 0.974 (0.011) \\
\hline
\end{tabular}}
\end{center}
\end{table}

\begin{figure}
    \centering
    \includegraphics[height = 19cm, width = 0.83\textwidth]{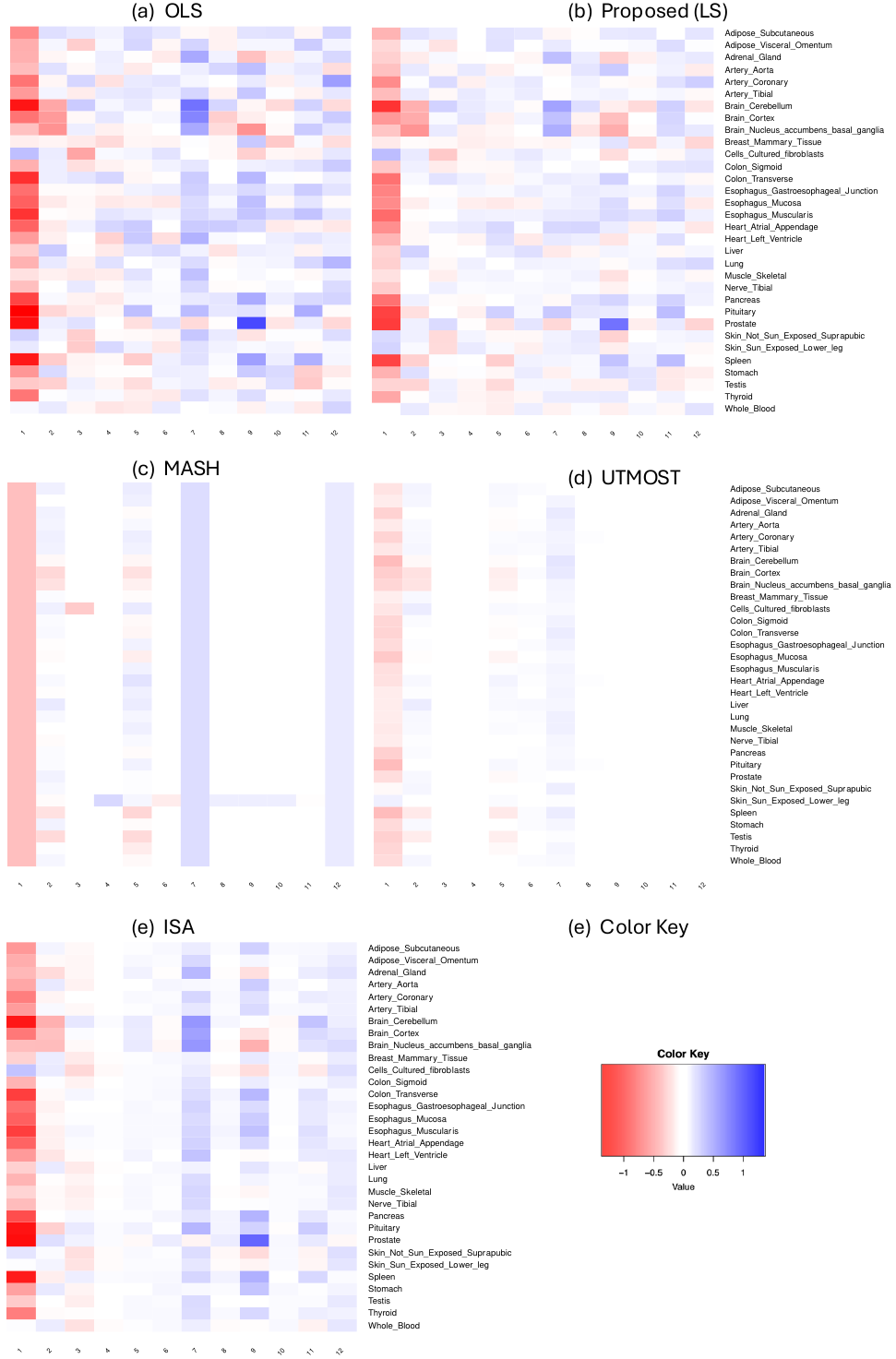}
    \caption{Heatmaps of estimated coefficients of $12$ cis-SNPs for the ZNF138 gene by different methods.}
    \label{heatmap}
\end{figure}

\section{Discussion}
We developed a formal data-integration method for multiple-source linear regressions with applications in TWAS studies. Existing literature views this problem through the lens of multi-response regression, although proposed solutions often make assumptions on the parameter structure that are not possible to verify in practice. We take an empirical Bayes approach here which naturally leads to shrinkage estimators, and are shown through extensive numerical experiments to have excellent performance under several parameter settings. Our work also builds upon the close connection between estimating the regression coefficient matrix and estimating covariance matrices. 

Addressing uncertainty in the estimates is a natural next step. This could be either handled directly or by computational methods such as bootstrap. Another interesting future direction is to consider nonparametric regressions with more predictive power within each source. 
To integrate predictions from different sources, the proposed empirical Bayes idea could still be useful.
In this paper, we only considered the low-dimensional cases where the number of subjects is larger than the number of predictors. In fact, high-dimensional cases are also important and could be another future direction, where we may leverage the asymptotic distribution of the de-biased lasso \citep{javanmard2014confidence} and incorporate more genotype biomarkers. Moreover, for large genetic datasets, summary statistics, such as correlations and standard errors, are easier to share than individual-level data due to privacy issues. Thus, it would be worthwhile to explore extending the proposed method to a summary statistics-based approach in the future.

{
Furthermore, in TWAS applications, certain domain knowledge can potentially be incorporated to further improve prediction and interpretability. For example, one may integrate information of known gene co-expression networks \citep{langfelder2008wgcna} or tissue similarity \citep{hu2019statistical} (e.g., from GTEx tissue correlation structure).
Specifically, suppose that $\Sigma_T^{n \times n}$ is a positive definite symmetric matrix with $\Sigma_T(i,j)$ encoding the similarity/distance between tissue $i$ and tissue $j$. Such a matrix is assumed to be constructed by the user. Then an observation model for the OLS estimates can be formed as $\widehat{B} \sim \text{MN}(B, \Sigma_T, Q)$ with the prior $B \sim \text{MN}(0,\Sigma_T, Q^{1/2} \Omega Q^{1/2})$. This scenario can be easily incorporated within our method since $\mathbb{E}(B \mid \widehat{B}) = \widehat{B} - \Sigma_T^{-1} \widehat{B} Q^{1/2}(\mathrm{I} + \Omega)^{-1} Q^{-1/2}$, which requires estimating the inverse of the marginal covariance matrix $(\mathrm{I} + \Omega)^{-1}$ under a rescaled relative savings loss (the rescaling is through $\Sigma_T^{-1}$).
}

\section*{Supplementary Material}
 Additional simulation results, proofs, and algorithms are provided in the supplementary material.

\section*{Acknowledgments}
This work was supported in part by the National Science Foundation under Grant DMS 2210860. 
We are also grateful to Yunlong Liu of Department of Medical and Molecular Genetics in Indiana University School of Medicine, for his helpful discussions on real data.

\bibliography{ref}
\bibliographystyle{chicago}
\clearpage
\onecolumn
\clearpage\pagebreak\newpage
\setcounter{secnumdepth}{3}
\setcounter{equation}{0}
\setcounter{page}{1}
\setcounter{table}{0}
\setcounter{section}{0}
\setcounter{subsection}{0}
\setcounter{figure}{0}
\setcounter{algorithm}{0}
\renewcommand{\theequation}{S.\arabic{equation}}
\renewcommand{\thesection}{S.\arabic{section}}
\renewcommand{\thepage}{S.\arabic{page}}
\renewcommand{\thetable}{S.\arabic{table}}
\renewcommand{\thefigure}{S.\arabic{figure}}
\renewcommand{\thealgorithm}{S.\arabic{algorithm}}
\section*{\uppercase{Supplementary Material}}
\section{Convergence of the loss}\label{supp_con_loss_sec}
We conduct a thorough simulation experiment to study the asymptotic approximation given by Theorem 5. Specifically, we consider three scenarios:
\begin{enumerate}
    \item {\bf Independence:} Suppose $X_{ij} \overset{iid}{\sim} \pi$ such that $\mathbb{E}(X_{ij}) = 0$ and $\mathbb{E}(X_{ij}^2) = 1$ and $\mathbb{E}(X_{ij}^{12}) < \infty$ for $i = 1, \ldots, n$, $j = 1, \ldots, p$. Here, the population covariance matrix is $\Sigma_n = \mathrm{I}_p$. This is satisfied by a large class of distributions, e.g. the exponential family. We set $\pi$ to be $\Gauss(0,1)$. Let $S_n = n^{-1}X^\T X$. For this setting, $F_n$ converges to the Marcenko-Pastur distribution when $p/n \to c \in (0,1)$. Moreover, it can be shown that $\mathbf{L}_{1, n} = \mathrm{tr}[(\widetilde{\Sigma}_n^{-1} - \Sigma_n^{-1})^2 S_n] \to 0$ as $n \to \infty$. Thus $\mathbf{L}_{1} = \lim_{n \to \infty} \mathbf{L}_{1,n} = 0$. We set $\widetilde{\Sigma}_n$ to be the proposed Stein shrinkage estimator from Theorem 6.

    \item {\bf Weak dependence:} Suppose $X_i^{p \times 1} \overset{iid}{\sim} \Gauss(0, \Sigma_n)$ for $i = 1, \ldots, n$ where $\Sigma_{n, jj} = 1 $ and $\Sigma_{n, jk} = \rho^{|j-k|}$ when $j \neq k$ for $|\rho|< 1$. This is an AR(1) structure. Unlike the previous case, here an analytical form of $\mathbf{L}_{1}$ is unknown although the limiting spectral distribution is known \citep{gray2006toeplitz}. We call it the weak dependence model since $\text{Cov}(X_{ij}, X_{ik}) \to 0$ as $|j-k| \to \infty$. For this case, we look at $\mathbf{L}_{1, n}$ for different choices of $n$. The estimator $\widetilde{\Sigma}_n$ is the same as the previous case.
    
    \item {\bf Strong dependence:} Suppose $X_i^{p \times 1} \overset{iid}{\sim} \Gauss(0, \Sigma_n)$ for $i = 1, \ldots, n$ where $\Sigma_n = \mathrm{I}_p + \rho \mathbf{1}_p\mathbf{1}_p'$. We set $\rho = 0.5$. This is a spike-covariance model. Here also, $\mathbf{L}_1$ is not available analytically. We note however that $F_n$ still converges to the Marcenko-Pastur law \citep{baik2006eigenvalues}. This is a strong dependence model since the $\text{Cov}(X_{ij}, X_{ik})$ is constant with respect to $|j-k|$. Similar to the previous setting, we look at the behaviour of $\mathbf{L}_{1, n}$ as $n$ increases. The same estimator $\widetilde{\Sigma}_n$ is also used here.
\end{enumerate}

\begin{figure}
   \begin{subfigure}{0.4\textwidth}
   \centering
    \includegraphics[width = 7.5cm, height = 4cm]{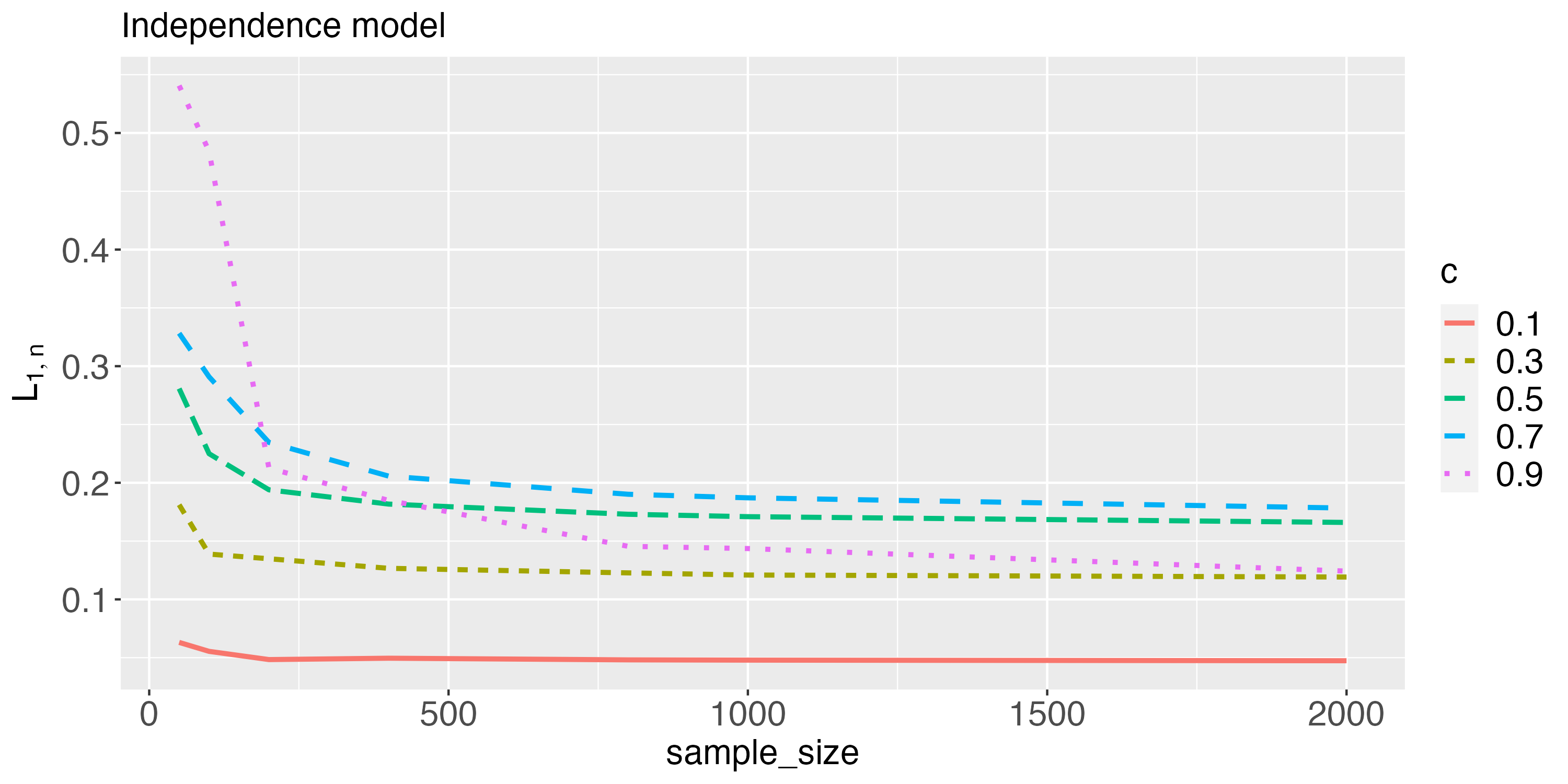}
   \end{subfigure}
   \begin{subfigure}{0.5\textwidth}
   \centering
       \includegraphics[width = 7.5cm, height = 4cm]{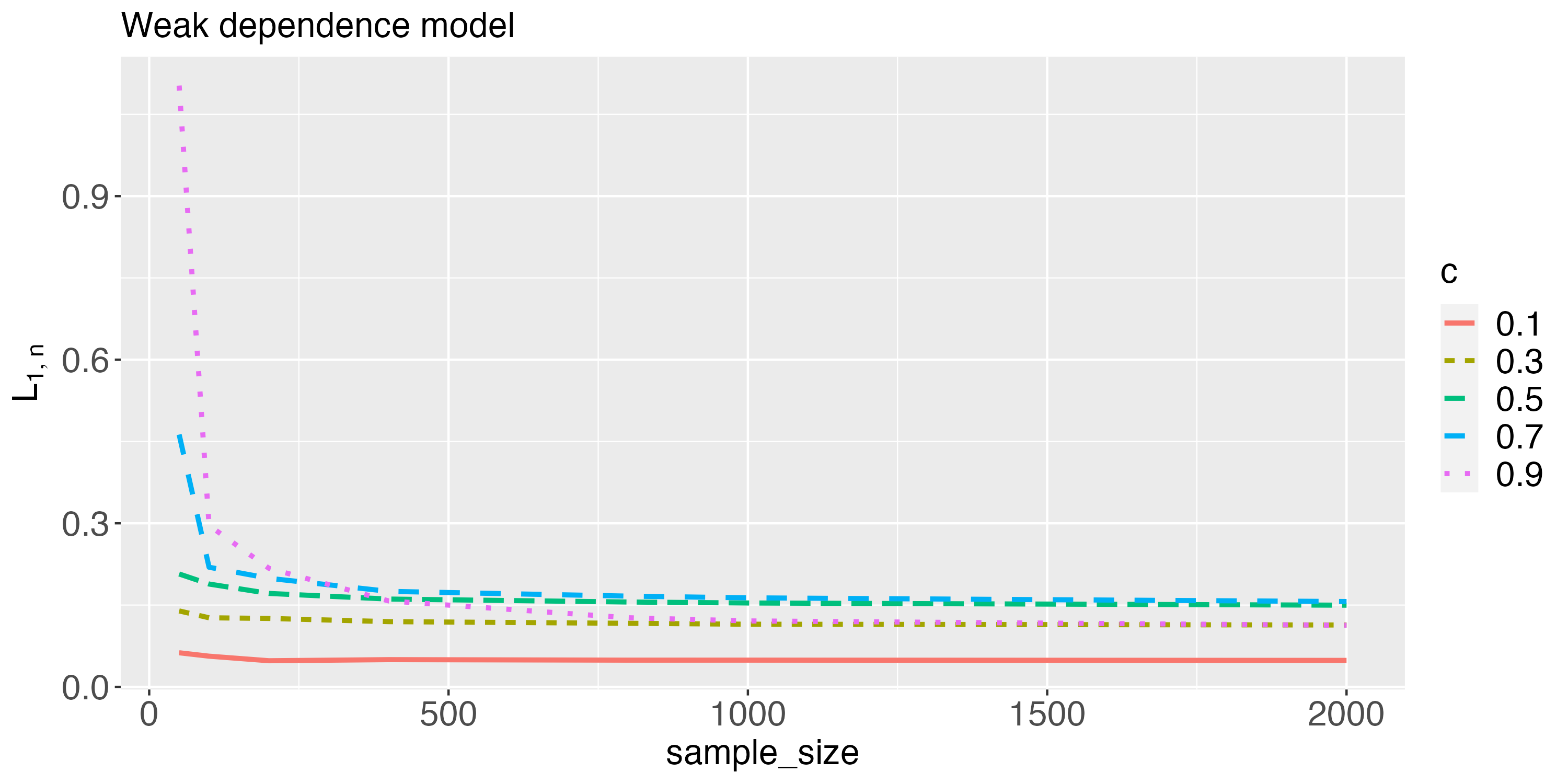}
   \end{subfigure}
   \begin{subfigure}{1\textwidth}
   \centering
       \includegraphics[width = 7.5cm, height = 4cm]{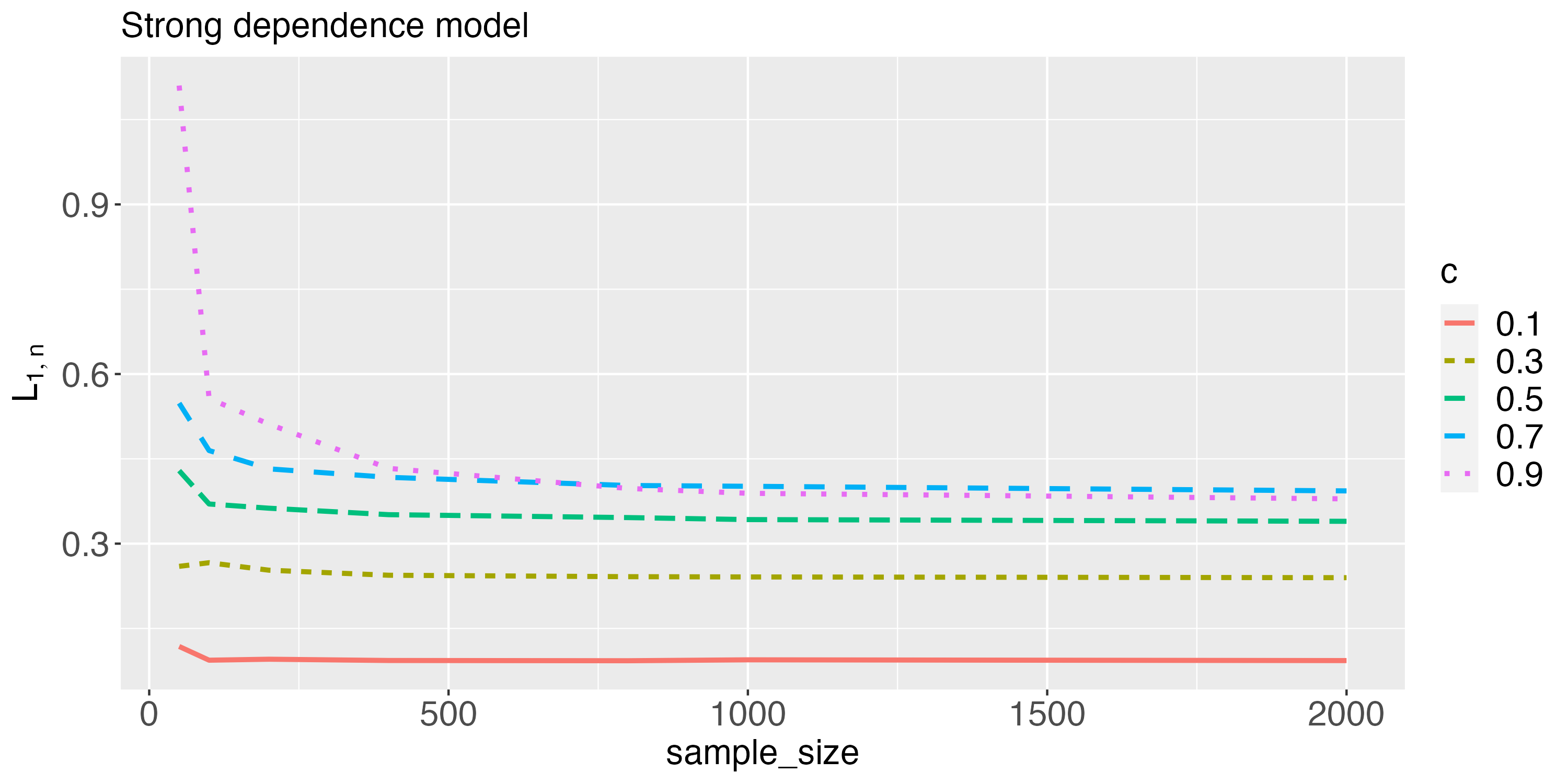}
   \end{subfigure}
   \caption{
   Behavior of $\mathbf{L}_{1,n}$ with $n$ for concentration ratios $c = 0.1 , 0.3 , 0.5 , 0.7 , 0.9$. The sample size $n = 50, 100, 200, 400, 800, 1000, 2000$. 
   }
   \label{fig:loss_convergence_main}
\end{figure}

For all these cases, we report the behavior of $\mathbf{L}_{1,n}$ with $n$ in Figure \ref{fig:loss_convergence_main} where we vary the concentration ratio $c$ from $0.1$ to $0.9$ with increments of $0.2$ and the sample size is varied within $\{50, 100, 200, 400, 800, 1000, 2000\}$. The results have been averaged over 100 replications. As can be seen from the results, {the difference between the empirical and the limiting loss is getting smaller for the {\bf Independence model} as $n$ grows, however, the convergence is very slow. This is perhaps not surprising given that almost sure convergence of even $F_n$ to $F$ can be as slow as $n^{-1/6}$ \citep[Theorem 8.23]{bai2010spectral}, especially when $c$ is close to 1. For the other two models, the empirical loss also seems to be converging.}

{
\section{Discussion on the risk of the proposed estimators}\label{supp_discuss_risk_sec}

In this section, we provide  a discussion on the risk of the proposed shrinkage rules. The posterior expectation $\mathbb{E}(\beta^{(t)} \mid \hat{\beta}^{(t)}) = (\mathrm{I} - C)\hat{\beta}^{(t)}$ where $C = Q^{1/2}(\mathrm{I} + \Omega)^{-1} Q^{-1/2}$, which is the basis of the linear shrinkage rule. Equivalently, in matrix notation we obtain the decision rule $\widetilde{B}(\widehat{B}) = \widehat{B} -\widehat{B}  [Q^{1/2}(\mathrm{I} + \Omega)Q^{1/2}]^{-1} Q $. To put things into a general framework, in the following discussion we use the matrix-variate normal distribution \cite{gupta2018matrix}. A matrix $X^{n\times p}$ is said to have a matrix-variate normal distribution with mean matrix $M^{n\times p} $ and covariance matrices $U^{n\times n}$, $V^{p \times p}$ (denoted as $X \sim \text{MN}(M, U, V)$) if its density function is of the form:
$$p(X \mid M, U, V) \propto \exp\left[-\frac{1}{2} \mathrm{tr}\left\lbrace V^{-1}(X - M)^\T U^{-1} (X-M) \right\rbrace \right].$$

Our current model and prior can be equivalently characterized as $\widehat{B} \mid B \sim \text{MN}(B, \mathrm{I}_n, Q)$ and $B \sim \text{MN}(0, \mathrm{I}_n, Q^{1/2}\Omega Q^{1/2})$, where $\widehat{B}$ and $B$ are $n \times p$ matrices with each row containing the vector $\hat{\beta}^{(t)}$ and $\beta^{(t)}$, respectively. The marginal distribution of $\widehat{B}$ is also a matrix-variate normal distribution: $\widehat{B} \sim \text{MN}(0, \mathrm{I}_n, Q^{1/2} (\mathrm{I} + \Omega) Q^{1/2})$. We write $m(\widehat{B})$ as the density of the marginal distribution. The linear shrinkage rule  $ \widetilde{B}(\widehat{B}) = \widehat{B} + [\nabla \log m(\widehat{B})]Q$. This is a generalization of Tweedie's formula \citep{efron2011tweedie} to the matrix-variate case. For the sequel, let us assume $Q = \mathrm{I}$. As long as $Q$ is known, this assumption does not lose generality. An adaptation of Stein's risk result \citep{stein1981estimation} for estimators $\widetilde{B}(\widehat{B}) = \widehat{B} + \nabla \log m(\widehat{B})$ in this context gives
$$\mathbf{R}(\widetilde{B}, B) = \mathbb{E}_{B}[\mathrm{tr}(\widetilde{B} - B)(\widetilde{B} - B)^\T] = np - \mathbb{E}_{B} [D \widetilde{B}(\widehat{B})],$$
where
$$D \widetilde{B}(\widehat{B}) = \norm{\nabla \log m(\widehat{B})}_F^2 - 2 \nabla^2 m(\widehat{B})/m(\widehat{B}), \quad \nabla^2 m(\widehat{B}) = \sum_{j,k} \nabla^2_{jk} m(\widehat{B}). $$
In the above expression, the expectation is taken with respect to the distribution of $\widehat{B}\mid B$. The equality for $\mathbf{R}(\widetilde{B}, B)$ quantifies the potential risk reduction for the estimator $\widetilde{B}$ compared to the maximum likelihood estimator, which has $np$. For example, if $\nabla^2 m(\widehat{B})<0$ for all $\widehat{B}$, then $\widetilde{B}$ strictly improves over the maximum likelihood estimator.

In the mixture setup, the marginal distribution of $B$ is $m_\star(\widehat{B}) = \sum_{k=1}^K \pi_k m_k(\widehat{B})$, where each $m_k(B)$ is the density of a matrix-variate normal distribution with mean 0, and covariance matrices $\mathrm{I}_n$, $Q^{1/2} \Omega_k Q^{1/2}$. The posterior mean also has a similar decomposition $\widetilde{B}_{mix}(\widehat{B}) = \widehat{B} + \sum_{k=1}^K \pi_k^\star(\widehat{B}) \nabla \log m_k(\widehat{B})$ with $\pi_k^\star(\widehat{B}) = \pi_k m_k(\widehat{B})/m_\star(\widehat{B})$, i.e. $\widetilde{B}_{mix}(\widehat{B}) = \sum_{k=1}^K \pi_k^\star(\widehat{B}) \widetilde{B}_k(\widehat{B})$. The risk of $\widetilde{B}_{mix}(\widehat{B})$ is then $\mathbf{R}(\widetilde{B}_{mix}, B) = np - \mathbb{E}_B[D \widetilde{B}_{mix}(\widehat{B})]$. This can be characterized following \cite[Corollary 3]{george1986minimax}, which gives
\begin{align*}
    D\widetilde{B}_{mix}(\widehat{B}) = \sum_{k=1}^K \pi_k^\star(\widehat{B})\left[ D \widetilde{B}_k (\widehat{B}) - \frac{1}{2} \sum_{l=1}^K \pi_l^\star (\widehat{B}) \norm{\widetilde{B}_k(\widehat{B}) - \widetilde{B}_l(\widehat{B})}_F^2\right]
\end{align*}
Thus, the risk gains of $\widetilde{B}_{mix}$ is a (posterior) weighted combination of risk gains of $D\widetilde{B}_k$ and a term that gives the shrinkage conflict between rules $\widetilde{B}_k$ and $\widetilde{B}_l$.
}

\section{Covariance estimation simulation}

In this section, our main focus is to compare the risk improvement in estimating a covariance matrix $\Sigma$ using the proposed estimator $\widetilde{\Sigma}(h) = U \Delta(h) U^\T$ where the observed sample covariance $S $ has the spectral decomposition $S = U \Lambda U^\T$. Here, $\Delta(h) = \text{diag}(\delta(h))$ and $\delta(h)$ is constructed following Theorem 6.
We choose $h$ according to \cite{ledoit2022quadratic} denoted as $h_0$, and the proposed data-dependent method developed in Section 3.1,
denoted as $\hat{h}$. Additionally, we consider an oracle choice of $h$ which is computed as follows. Given a range of values of $h$, we compute $\widetilde{\Sigma}(h)$, and compute the corresponding risk $\mathbb{E}[\mathbf{L}(\Sigma^{-1}, \widetilde{\Sigma}(h)^{-1})]$ approximated by taking Monte Carlo averages over 100 replications. We then select the oracle $h$ for which the risk is minimum. We write the oracle choice of $h$ as $\underline{h}$ and the resulting estimate of $\Sigma$ as $\underline{\Sigma} = \widetilde{\Sigma}(\underline{h})$. Clearly, this selection of $h$ requires knowledge of the true $\Sigma$.
This forms the baseline of risk improvement that we would consider while computing the percentage relative improvement in average loss (PRIAL) defined as
\begin{equation}
    \text{PRIAL}[\Sigma(h)] = \dfrac{\mathbb{E}[\mathbf{L}(\Sigma^{-1}, S^{-1})] - \mathbb{E}[\mathbf{L}(\Sigma^{-1}, \widetilde{\Sigma}(h)^{-1}]}{\mathbb{E}[\mathbf{L}(\Sigma^{-1}, S^{-1})] - \mathbb{E}[\mathbf{L}(\Sigma^{-1}, \underline{\Sigma}^{-1}]} \times 100\%
\end{equation}
We report the PRIAL as it varies over the concentration ratio $p/n$ over the interval $\{0.1, 0.2,\\ \ldots, 0.9\}$. Here, the product $np$ is fixed at 20000, and then $(n, p)$ are chosen such that the concentration ratio is closest to elements of the set above. The true covariance matrix in all these cases were assumed to have a factor structure, i.e. $\Sigma = \Xi \Xi^\T + \mathrm{I}_p$ where $\Xi \in \mathbb{R}^{p \times k}$. We set $k = 5$ for all cases and generate elements of $\Xi$ independently from $\Gauss(0,1)$. All expectations are approximated by 100 independent Monte Carlo replications.
\begin{figure}
    \centering
    \includegraphics[width=0.8\textwidth, height = 6cm]{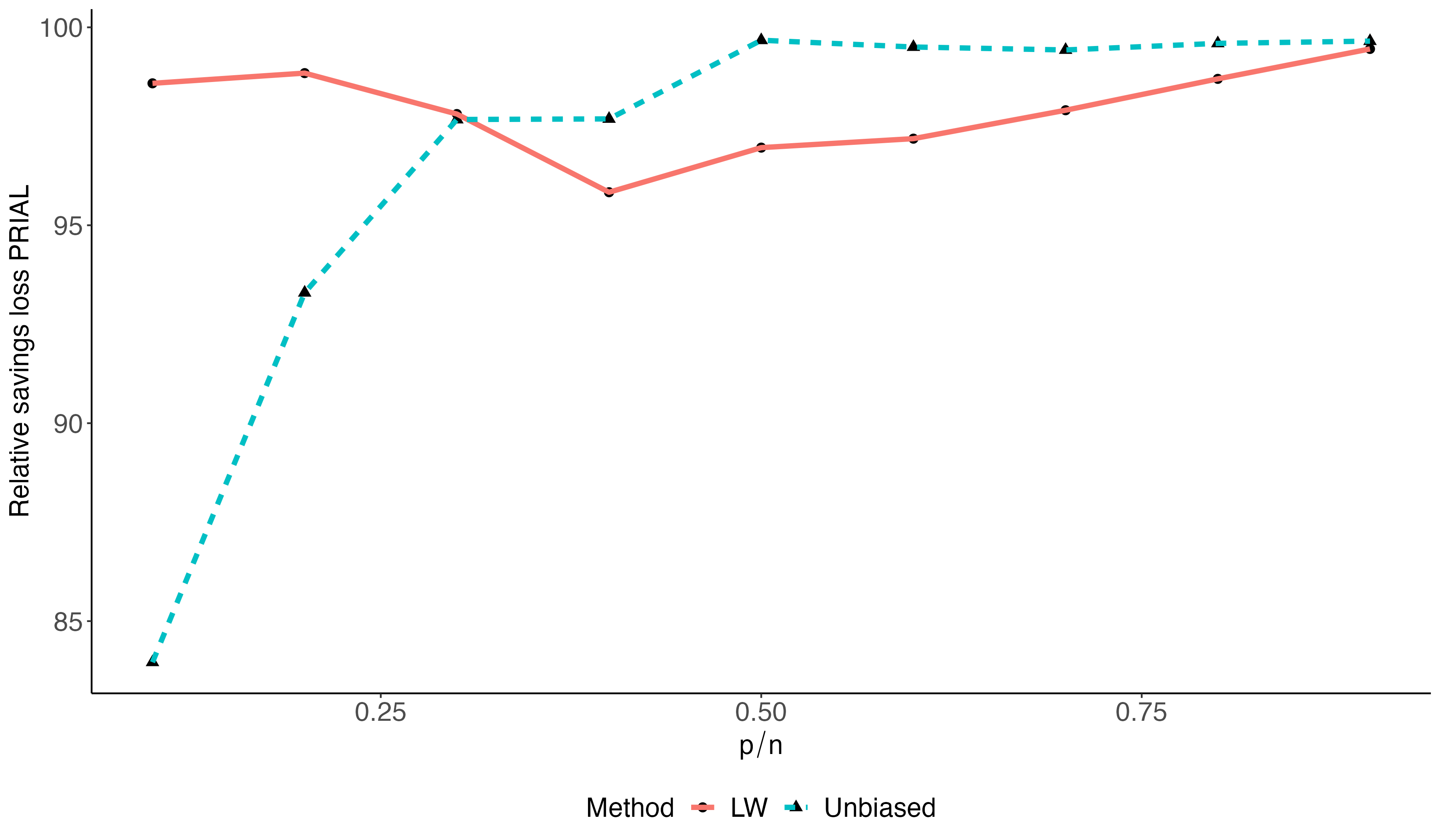}
    \caption{Relative savings loss PRIAL with $h$ chosen following \cite{ledoit2022quadratic} (LW) and the unbiased risk estimate procedure (unbiased) when the concentration ratio $p/n$ varies over the interval $\{0.1, 0.2, \ldots, 0.9\}$.}
    \label{fig:PRIAL}
\end{figure}

The result is summarized in Figure \ref{fig:PRIAL}, which shows that for smaller concentration ratios, the default choice suggested by Ledoit-Wolf gives better results compared to the data-dependent method developed here. However, the benefit of the data-dependent method becomes apparent as we move into cases with a larger $p/n$,
where the risk improvements are significant.

{
\section{Simulations for  high-dimensional settings}\label{sec:supp_sim_high_sec}

In this section, we extend our experiments to settings where the number of tissues/sources is $n = 200$ and $p = 100, 120, 150, 500$. The number of data points within each tissue/source is $N = 1000$ when $(n , p) = (200, 500)$. For all other cases, we set $N = 200$. This is to ensure that the OLS estimator exists for all of these settings. For these high-dimensional examples, we compared our linear shrinkage estimator with default choice of $h$ (LS), the linear shrinkage estimator with the data-dependent choice of $h$ (ULS) versus other methods - UTMOST, ISA. 
We do not report the proposed local linear shrinkage estimator as the results were almost identical with LS/ULS. Also, the existing implementation of MASH frequently ran into convergence issues in these high-dimensional settings. Hence, we do not report results for this method as well. Finally, we note that for the case $(n,p) = (200, 500)$, the ULS estimator does not work since it assumes $p<n$. These results are reported in Tables \ref{tab:MSE_high_dim_1_main}, \ref{tab:MSE_high_dim_2_main}, \ref{tab:PE_high_dim_1_main}, and \ref{tab:PE_high_dim_2_main}. Our conclusions from these new experiments remain consistent from our previous experiments - the proposed method performs significantly better across all different settings. 

\begin{table}[]
    \centering
   \scalebox{0.8}{
    \begin{tabular}{c|c|cccc|cccc}
    \hline
         & & \multicolumn{4}{c}{Low-rank} & \multicolumn{4}{c}{Approximately sparse} \\
         \hline
        & &  LS & ULS & UTMOST & ISA &  LS & ULS & UTMOST & ISA \\
        \hline
          & $\rho = 0$ & 0.0009 & 0.0008 & 2.89 & 0.0003 & 0.002 & 0.002 & 0.013 & 0.039 \\
         \multirow{3}{*}{$p = 100$} & $ \rho = 0.5$ & 0.001 & 0.001 & 2.83 & 0.0006 & 0.004 & 0.004 & 0.015 & 0.04 \\
                   & $\rho = 0.8$ & 0.004 & 0.004 & 4.79 & 0.001 & 0.01 & 0.01 & 0.03 & 0.04 \\
                   \hline
        & $\rho = 0$ & 0.001 & 0.0009 & 2.99 & 0.0003 & 0.002 & 0.002 & 0.015 & 0.039 \\
         \multirow{3}{*}{$p = 120$} & $ \rho = 0.5$ & 0.002 & 0.001 & 3.16 & 0.0007 & 0.004 & 0.004 & 0.015 & 0.041   \\
                   & $\rho = 0.8$ &  0.005 & 0.004 & 5.08 & 0.001 & 0.01 & 0.01 & 0.03 & 0.04 \\
                   \hline
         & $\rho = 0$ & 0.001 & 0.001 & 2.63 & 0.0006 & 0.002 & 0.002 & 0.015 & 0.04 \\
         \multirow{3}{*}{$p = 150$} & $ \rho = 0.5$ &  0.002 & 0.002 & 3.55 & 0.001 & 0.005 & 0.005 & 0.016 & 0.045  \\
                   & $\rho = 0.8$ & 0.005 & 0.005 & 5.79 & 0.002 & 0.012 & 0.012 & 0.031 & 0.049\\
                   \hline
        & $\rho = 0$ & 0.001 & - & 2.78 & 0.0002 & 0.002 & - & 0.016 & 0.07 \\
         \multirow{3}{*}{$p = 500$}& $ \rho = 0.5$ & 0.002 & - & 4.90 & 0.0003 & 
         0.004 & - & 0.027 & 0.09\\
                   & $\rho = 0.8$ & 0.005 & - & 6.85 & 0.001 & 0.009 & - & 0.036 & 0.09 \\
                   \hline
    \end{tabular}
  }
    \caption{MSE of different estimators over 20 replications in high-dimensional settings.}
    \label{tab:MSE_high_dim_1_main}
\end{table}

\begin{table}[]
    \centering
   \scalebox{0.8}{
    \begin{tabular}{c|c|cccc|cccc}
    \hline
         & & \multicolumn{4}{c}{Horseshoe} & \multicolumn{4}{c}{Mixture} \\
         \hline
        & &  LS & ULS & UTMOST & ISA & LS & ULS & UTMOST & ISA \\
        \hline
          & $\rho = 0$  & 0.002 & 0.002 & 13.01 & 14.57 & 0.002 & 0.002 & 23.93 & 73.27\\
         \multirow{3}{*}{$p = 100$} & $ \rho = 0.5$ &  0.004 & 0.004 & 14. 21 & 11.54 & 0.005 & 0.005 & 22.86 & 71.60\\
                   & $\rho = 0.8$ &  0.01 & 0.01 & 16.88 & 13.78 & 0.01 & 0.01 & 28.37 & 71.05\\
                   \hline
        & $\rho = 0$ &  0.002 & 0.002 & 17.42 & 15.83 & 0.002 & 0.002 & 23.03 & 68.91\\
         \multirow{3}{*}{$p = 120$} & $ \rho = 0.5$ &  0.005 & 0.005 & 18.21 & 15.21 & 0.005 & 0.005 & 24.01 & 69.68  \\
                   & $\rho = 0.8$ &  0.013 & 22.81 & 16.37 & 0.012 & 0.012 & 25.77 & 72.31 \\
                   \hline
         & $\rho = 0$ &  0.003 & 0.003 & 15.21 & 13.87 & 0.01 & 0.01 & 22.71 & 53.65\\
         \multirow{3}{*}{$p = 150$} & $ \rho = 0.5$ &   0.005 & 0.005 & 15.52 & 14.50 & 0.005 & 0.005 & 24.93 & 57.90 \\
                   & $\rho = 0.8$ & 0.01 & 0.01 & 16.90 & 14.66 & 0.01 & 0.01 & 28.72 & 65.22\\
                   \hline
        & $\rho = 0$ &  0.003 & - & 150.01 & 45.04 & 0.002 & - & 25.03 & 67.32\\
         \multirow{3}{*}{$p = 500$}& $ \rho = 0.5$ & 0.004 & - & 153.77 & 44.39 & 0.004 & - &  48.57 & 74.01\\
                   & $\rho = 0.8$  & 0.01 & - & 152.87 & 59.31 & 0.005 & - & 53.97 & 76.29\\
                   \hline
    \end{tabular}
   }
    \caption{MSE of different estimators over 20 replications in high-dimensional settings.}
    \label{tab:MSE_high_dim_2_main}
\end{table}

\begin{table}[]
    \centering
   \scalebox{0.8}{
    \begin{tabular}{c|c|cccc|cccc}
    \hline
         & & \multicolumn{4}{c}{Low-rank} & \multicolumn{4}{c}{Approximately sparse} \\
         \hline
        & &  LS & ULS & UTMOST & ISA &  LS & ULS & UTMOST & ISA  \\
        \hline
          & $\rho = 0$ & 1.09 & 1.08 & 146.92 & 1.01 & 1.22 & 1.20 & 2.43 & 5.01  \\
         \multirow{3}{*}{$p = 100$} & $ \rho = 0.5$ &1.07 & 1.05 & 149.99 & 1.01 & 1.27 & 1.25 & 1.77 & 6.35  \\
                   & $\rho = 0.8$ & 1.12 & 1.08 & 86.27 & 1.24 & 1.20 & 1.14 & 1.54 & 5.54 \\
        \hline
        & $\rho = 0$ & 1.07 & 1.06 & 436.4 & 1.01 & 1.33 & 1.32 & 2.72 & 5.88 \\
         \multirow{3}{*}{$p = 120$} & $ \rho = 0.5$ & 1.12 & 1.06 & 189.70 & 1.06 & 1.20 & 1.17 & 1.73 & 4.85  \\
                   & $\rho = 0.8$ & 1.16 & 1.12 & 133.76 & 1.22 & 1.28 & 1.26 & 1.85 & 7.41 \\
        \hline
         & $\rho = 0$ & 1.15 & 1.13 & 167.41 & 1.06 & 1.43 & 1.35 & 3.33 & 7.10 \\
         \multirow{3}{*}{$p = 150$} & $ \rho = 0.5$ & 1.21 & 1.05 & 183.3 & 1.11 & 1.48 & 1.43 & 2.33 & 7.14 \\
                   & $\rho = 0.8$ & 1.15 & 1.12 & 182.83 & 1.10 & 1.47 & 1.40 & 2.05 & 5.42 \\
                   \hline
        & $\rho = 0$ & 1.47 & - & 713.4 & 1.08 & 1.93 & - & 8.60 & 20.17 \\
         \multirow{3}{*}{$p = 500$}& $ \rho = 0.5$ & 1.57 & - & 881.4 & 1.12 & 1.96 & - & 8.03 & 16.94 \\
                   & $\rho = 0.8$ & 1.47 & - & 525.2 & 1.20 & 1.91 & - & 4.65 & 36.13 \\
                   \hline
    \end{tabular}
   }
    \caption{{PE} of different estimators over 20 replications in high-dimensional settings.}
    \label{tab:PE_high_dim_1_main}
\end{table}

\begin{table}[]
    \centering
   \scalebox{0.8}{
    \begin{tabular}{c|c|cccc|cccc}
    \hline
         & & \multicolumn{4}{c}{Horseshoe} & \multicolumn{4}{c}{Mixture} \\
         \hline
        & &  LS & ULS & UTMOST & ISA &  LS & ULS & UTMOST & ISA  \\
        \hline
          & $\rho = 0$ &  1.23 & 1.20 & 406.1 & 1115.2 & 1.28 & 1.21 & 1072.9 & 4938.9 \\
         \multirow{3}{*}{$p = 100$} & $ \rho = 0.5$ & 1.26 & 1.19 & 401.7 & 1025.5 & 1.24 & 1.07 & 1158.5 & 5877.7  \\
                   & $\rho = 0.8$ &  1.26 & 1.24 & 352.2 & 1067.2 & 1.23 & 1.17 & 953.2 & 6744.6\\
                   \hline
        & $\rho = 0$ &  1.29 & 1.25 & 652.8 & 1724.1 & 1.39 & 1.33 & 2921.2 & 8493.4\\
         \multirow{3}{*}{$p = 120$} & $ \rho = 0.5$ &  1.35 & 1.27 & 768.1 & 1672.5 & 1.29 & 1.12 & 1402.7 &  9736.2 \\
                   & $\rho = 0.8$ &  1.31 & 1.29 & 807.3 & 1227.7 & 1.30 & 1.26 & 1276.3 & 8506.1\\
        \hline
         & $\rho = 0$ &  1.47 & 1.42 & 1382.8 & 2207.8 & 1.39 & 1.33 & 1610.5 & 8432.9\\
         \multirow{3}{*}{$p = 150$} & $ \rho = 0.5$ &  1.53 & 1.50 & 1872.2 & 1483.6 & 1.41 & 1.32 & 1849.1 & 7992.8\\
                   & $\rho = 0.8$ &  1.55 & 1.52 & 1923.5 & 1973.1 & 1.46 & 1.44 & 1689.6 & 8112.6\\
                   \hline
        & $\rho = 0$ &  2.07 & - & 5 $\times 10^4$ & 23328.9 & 1.93 & - & 11913.1 & 31440.7\\
         \multirow{3}{*}{$p = 500$}& $ \rho = 0.5$ &  2.43 & - & 5 $\times 10^4$ & 20461.4 & 1.96 & - & 12135.4 & 34729.4\\
                   & $\rho = 0.8$ &  2.63 & - & $5 \times 10^4$ & 22646.5 & 1.99 & - & 12282.4 & 33947.3\\
                   \hline
    \end{tabular}
   }
    \caption{{PE} of different estimators over 20 replications in high-dimensional settings.}
    \label{tab:PE_high_dim_2_main}
\end{table}

}

{
\section{Analysis of Yeast Cell Cycle dataset}\label{supp_real_app_sec}

 In this section, we apply the proposed LS method to the Yeast Cell Cycle dataset \cite{chun2010sparse},
and compare it with the existing ordinary least squares (OLS),  the Unified Test for MOlecular SignaTures (UTMOST) \citep{hu2019statistical}, the Iterated stable autoencoder (ISA) \citep{josse2016bootstrap}, and the Multivariate Adaptive Shrinkage (MASH) \citep{kim2024flexible} methods.
The Yeast Cell Cycle dataset contains $18$ responses and $106$ covariates for $542$ genes (i.e., sample size $N=542$). 
Each of the responses 
represents mRNA levels measured at every 7 minutes during
119 minutes. The covariates consist of the binding information for $106$ transcription
factors.

We conduct a $10$-fold cross-validation analysis for each response to evaluate the prediction performance of each method.
Specifically, for each response, we randomly split all the observed samples into $10$ equally sized folds, and name them Folds $1-10$. For each $i = 1,\dots,10$, we treat Fold $i$ in all the tissues as a testing set, and the remaining folds in all the tissues together as a training set. We then use the average prediction mean squared error (PMSE) across all the folds and responses for the evaluation of prediction accuracy.

The average PMSEs of all the methods are provided in Table \ref{Table_PMSE_yeast}. The results show that the proposed method outperforms all the existing methods in terms of the average PMSE. We also provide the PMSE of each response in the left plot of Figure \ref{yeast_plots}, which shows that the proposed method produces the smallest PMSE among all the methods for most of responses.
In addition, we calculate the Pearson correlation between the predicted values and true values for each response and each method, and present the results in the right plot of Figure \ref{yeast_plots}.
We can observe that the correlation corresponding to the proposed method is higher than those of other methods for most responses.

In summary, the proposed method works well in the multi-response regression for the Yeast Cell Cycle dataset  which is not related to TWAS.

\begin{table}[hbt!]
\begin{center}
\vspace{2mm}
\resizebox{0.9\columnwidth}{!}{%
\begin{tabular}{l|ccccc}
\hline
Methods & Proposed (LS) & OLS & UTMOST & ISA & MASH \\ 
\hline
PMSE & 0.164 (0.095) & 0.216 (0.103)  & 0.188 (0.091) & 0.171 (0.089) & 0.188 (0.087) \\
\hline
\end{tabular}}
\end{center}
\vspace{-5mm}
\caption{Average prediction mean squared errors (PMSEs) for different methods, with standard deviation (SD) in the parentheses. 
``PMSE'' represents the average PMSE across all the 18 responses.
}\label{Table_PMSE_yeast}
\end{table}

\begin{figure}
    \centering
    \includegraphics[height = 7cm, width = 0.45\textwidth]{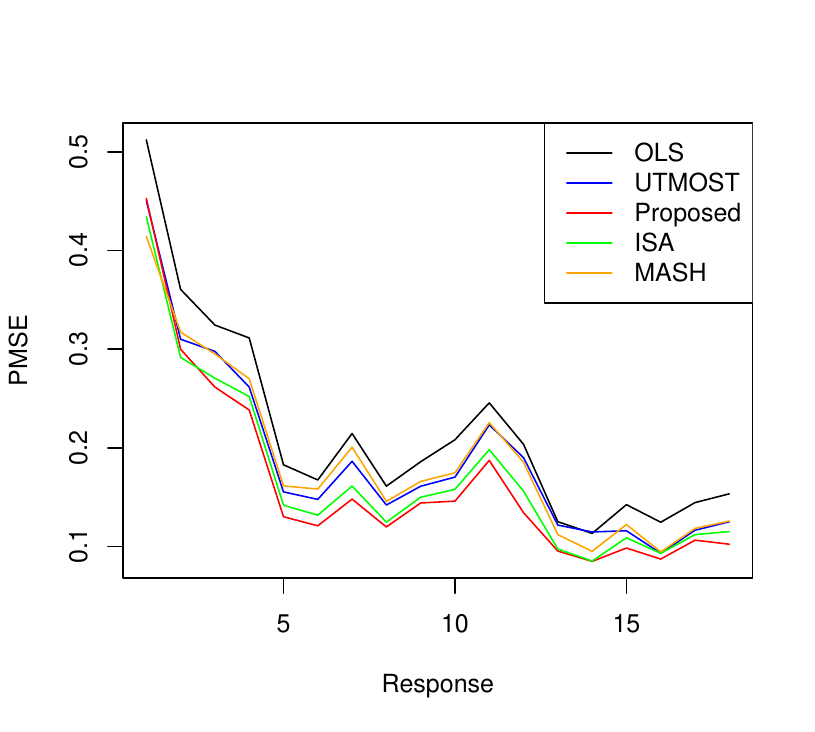}
     \includegraphics[height = 7cm, width = 0.45\textwidth]{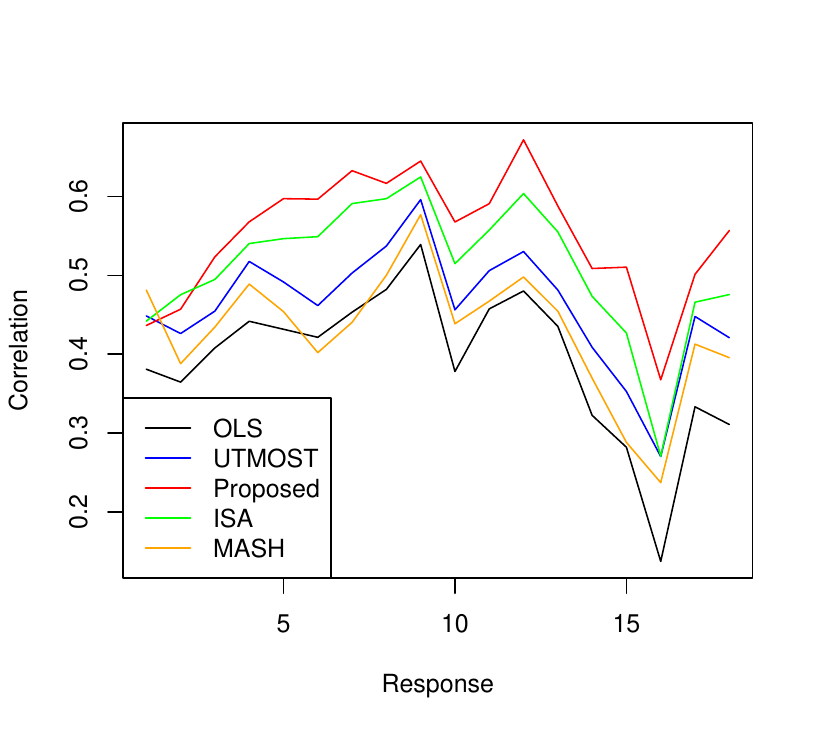}
    \caption{Prediction mean squared error (PMSE) and correlation for each response and each the method.}
    \label{yeast_plots}
\end{figure}

}

\section{Proofs of Sections 2
and 3 
}
\subsection{Proof of Proposition 1
}
\label{prop:equivalence_proof}
Without loss of generality, set $Q = \mathrm{I}$. Under the assumed model $\mathbb{E}(B \mid \widehat{B}) = \widehat{B}(\mathrm{I} - C) $. Thus,
\begin{align*}
&\mathbb{E}[\mathbf{L}(B, \widetilde{B})] = \mathbb{E}_{\widehat{B}}\mathbb{E}_{B\mid \widehat{B}}[\mathbf{L}(B, \widetilde{B} )]\\
= &\mathbb{E}_{\widehat{B}}\mathbb{E}_{B \mid \widehat{B}}\tr [\{\widehat{B} (\mathrm{I} -  \widetilde{C})- B\}\{\widehat{B} (\mathrm{I} -  \widetilde{C})- B\}^\T] \\
=& \mathbb{E}_{\widehat{B}} \mathbb{E}_{B\mid \widehat{B}}[\tr (B B^\T)] 
- \mathbb{E}_{\widehat{B}}\mathbb{E}_{B \mid \widehat{B}} \left[\tr \{B  (I - \widetilde{C}^\T ) \widehat{B}^\T)\}\right] - \mathbb{E}_{\widehat{B}}\mathbb{E}_{B\mid \widehat{B}} \left[\tr \{ \widehat{B}(\mathrm{I} -  \widetilde{C})  B^\T\} \right] + \\
& +  \mathbb{E}_{\widehat{B}} \mathbb{E}_{B\mid \widehat{B}} [\tr \{ \widehat{B}(\mathrm{I} -  \widetilde{C})  (\mathrm{I} -    \widetilde{C}^\T) \widehat{B}^\T\}] \\
=& \mathbb{E}_{\widehat{B}} \mathbb{E}_{B\mid \widehat{B}}[\tr (B B^\T)]
- \mathbb{E}_{\widehat{B}}  \left[\tr \{\widehat{B}(\mathrm{I} - C ) (\mathrm{I} - \widetilde{C}^\T)\widehat{B}^\T\}\right] - \mathbb{E}_{\widehat{B}}  \left[\tr \{\widehat{B}(\mathrm{I} - \widetilde{C} ) (\mathrm{I} - C^\T)\widehat{B}^\T\}\right]\\
+ & \mathbb{E}_{\widehat{B}}  [\tr \{ \widehat{B}(\mathrm{I} -  \widetilde{C}) (\mathrm{I} -    \widetilde{C}^\T) \widehat{B}^\T\}] .
\end{align*}
Since $\mathbb{E}_{\widehat{B}} \mathbb{E}_{B\mid \widehat{B}} [\tr (B B^\T)] = \mathbb{E}_{\widehat{B}}\left\{\mathbb{E}_{B\mid \widehat{B}}\left[\sum_{t=1}^T  {\beta^{(t)}}^\T  \beta^{(t)} \right]\right\}$, and 
\begin{align*}
\mathbb{E}_{B\mid \widehat{B}} [\tr (B\Psi B^\T)] 
=& T \tr[ (\mathrm{I} - C^\T)] + \mathbb{E}_{\widehat{B}} [ \tr (\widehat{B}(\mathrm{I} -  C) (\mathrm{I} - C^\T)\widehat{B}^\T)],
\end{align*}
we have
\begin{align*}
\mathbb{E}_{\widehat{B}} \mathbb{E}_{B\mid \widehat{B}}[\mathbf{L}(B, \widetilde{B})] &= \mathbb{E}_{\widehat{B}} [ \tr (\widehat{B}(\mathrm{I} -  C) (\mathrm{I} - C^\T)\widehat{B}^\T)] - \mathbb{E}_{\widehat{B}}  \left[\tr \{\widehat{B}(\mathrm{I} - C ) (\mathrm{I} - \widetilde{C}^\T)\widehat{B}^\T\}\right] \\
- \mathbb{E}_{\widehat{B}}  &\left[\tr \{\widehat{B}(\mathrm{I} - \widetilde{C} ) (\mathrm{I} - C^\T)\widehat{B}^\T\}\right] + \mathbb{E}_{\widehat{B}}  [\tr \{ \widehat{B}(\mathrm{I} -  \widetilde{C})  (\mathrm{I} -    \widetilde{C}^\T) \widehat{B}^\T\}] +\text{constant}\\
&=\mathbb{E}_{\widehat{B}}[\tr \{(\mathrm{I} - C) - (\mathrm{I} - \widetilde{C})\}  \{(\mathrm{I} - C) - (\mathrm{I} - \widetilde{C})\}^\T \widehat{B}^\T \widehat{B}] + \text{constant}\\
& = \mathbb{E}_{\widehat{B}}[\tr \{(\widetilde{C} - C) (\widetilde{C} - C)^\T\}\widehat{B}^\T \widehat{B}]+ \text{constant}\\
&= \mathbb{E}_{\widehat{B}}[\tr \{(\widetilde{\Sigma}^{-1} - \Sigma^{-1})^2 \}\widehat{B}^\T \widehat{B}] +\text{constant},
\end{align*}
which was to show.

\subsection{Proof of Theorem 5 
}
Recall that
    \begin{align*}
    \mathbf{L}_{m,n}(\Sigma_n^{-1}, \tilde{\Sigma}_n^{-1}; \mathrm{I}) 
    & =  \int_{-\infty}^{\infty} x^m d\Phi_n^{(-2)} (x) - 2 \int_{-\infty}^{\infty} \frac{x^m}{\delta_{n}(x)} d\Phi_{n}^{(-1)}(x) + \int_{-\infty}^{\infty} \frac{x^m}{\delta^2_n(x)} dF_n(x).
\end{align*}
Since $x^m$ is a continuous function and $\Phi_n^{(-2)} (x)$ converges weakly to $\Phi^{(-2)} (x)$ by Lemma 1,
\begin{align*}
    \int_{-\infty}^{\infty} x^m d\Phi_n^{(-2)} (x) \overset{a.s.}{\to} \int_{-\infty}^{\infty} x^m d\Phi^{(-2)} (x).
\end{align*}
Assumption 4
and the continuous mapping theorem imply that 
\begin{align*}
    \frac{x^m}{\delta_{n}(x)} \overset{a.s.}{\to} \frac{x^m}{\delta(x)} \quad \text{ and } \quad 
    \frac{x^m}{\delta_{n}^2(x)} \overset{a.s.}{\to} \frac{x^m}{\delta^2(x)}
\end{align*}
for $x\in\supp (F)$.
 In addition, the convergence is {uniform} for $x\in \cup_{k=1}^K [a_k+\eta, b_k-\eta]$ for any small $\eta>0$.
    Furthermore, there exists a finite nonrandom constant $\widetilde{M}$  such that $|x^m/\delta_n(x)|$ and $|x^m/\delta_n^2(x)|$ are uniformly bounded 
{from above} by $\widetilde{M}$ almost surely for all $x\in \cup_{k=1}^K [a_k-\eta, b_k+\eta]$, large $n$, and small $\eta>0$.

By \citep[Lemma 11.1]{ledoit2018optimal}, under our working assumptions, $\Phi_n^{(-1)} (x)$ converges weakly to $\Phi^{(-1)} (x)$, $\Phi^{(-1)} (x)$ is continuously differentiable on $\mathbb{R}$, and  
    \begin{align*}
        \Phi^{(-1)}(x) = \int_{-\infty}^x \phi^{(-1)}(\xi) dF(\xi),
    \end{align*}
for $\forall x \in \mathbb{R}$. Note that, by \cite{silverstein1995analysis}, \cite{silverstein1995empirical}, and \cite{silverstein1995strong}, we also have
\begin{align*}
    F_n(x) \overset{a.s.}{\to} F(x) \quad \forall x\in\mathbb{R},
\end{align*}
and $F(x)$ is continuously differentiable.
{
Thus, similar to the proofs of \citep[Lemma 11.2]{ledoit2018optimal}, we can have
\begin{align*}
    \int_{-\infty}^{\infty} \frac{x^m}{\delta_{n}(x)} d\Phi_{n}^{(-1)}(x) \overset{a.s.}{\to} \sum_{k=1}^K \int_{a_k}^{b_k} \frac{x^m}{\delta(x)}\phi^{(-1)}(x) dF(x), 
\end{align*}
and 
\begin{align*}
        \int_{-\infty}^{\infty} \frac{x^m}{\delta^2_n(x)} dF_n(x) \overset{a.s.}{\to}     \sum_{k=1}^K \int_{a_k}^{b_k} \frac{x^m}{\delta^2(x)} dF(x).
\end{align*}
\subsection{Proof of Corollary 1 
}
To find a function $\delta (x)$ that minimizes the limit in Equation (3.4),
for each fixed $x$, we take derivative of 
\begin{align*}
-2\frac{x^m}{\delta(x)}\phi^{(-1)}(x)  
    + \frac{x^m}{\delta^2(x)}    
\end{align*}
with respect to $\delta$, and let it equal zero. Here we do not consider the first term in Equation (3.4)
since it does not involve $\delta(x)$. 
Then, the minimizer is $1/\phi^{(-1)}(x)$.
\subsection{Proof of Theorem 6
}
By the proof of \cite[Theorem 3.1]{ledoit2022quadratic}, we obtain $\delta_{n}^* (x) \overset{p}{\to} \delta^*(x)$ for any $x \in \supp (F)$.
}
\section{Proof of Theorem 7
}
We record the following results from \cite{haddouche2021scale,boukehil2021estimation} which will be useful in proving Theorem 7. 
\begin{lemma}\label{lemma:stein_haff}
    Suppose $Q \sim \mathrm{W}_p(\Sigma, n), \, n>p$ and $G(Q)$ is a $p\times p$ weakly differentiable matrix function. If $\mathbb{E}_{\Sigma}\left[\mid \mathrm{tr}\{\Sigma Q G(Q)\}\mid\right] < \infty$, then 
    $$\mathbb{E}_{\Sigma}\left[ \mathrm{tr}\{\Sigma^{-1} Q G(Q)\}\right] = \mathbb{E}_{\Sigma} \left[ \mathrm{tr}\left\lbrace (n-p - 1) G(Q) + 2 D_{Q}(G(Q)^\T Q) \right\rbrace\right],$$
    where $D_{Q}$ is a differential operator defined as $D_Q = \left\lbrace \frac{1}{2}(1 + d_{ij}) \frac{\partial}{\partial Q_{ij}}\right\rbrace$ for $1\leq i,j\leq p$ with $d_{ij} = 1$ if $i =j$ and $0$ otherwise.
\end{lemma}
\begin{lemma}\label{lemma:D_operator}
    Suppose $S = P L P^\T$ and $G(S) = P \Psi(L) P^\T$ are the spectral decomposition of S and $G(S)$, respectively for symmetric positive definite $S$. Then,
    $$D_{S}\{G(S)\} = P \Psi^{(1)} P^\T + \frac{1}{2} \mathrm{tr}\{L^{-1}\Psi(L)\} (\mathrm{I}_p - PP^\T),$$
    where the $j$-th element of the diagonal matrix $\Psi^{(1)} $ is $\psi_j^{(1)} = \frac{\partial \psi_i}{\partial l_i} + \frac{1}{2} \sum_{i\neq j}^{p} \dfrac{\psi_j - \psi_i}{l_j - l_i}$.
\end{lemma} 
\noindent We are now ready to prove Theorem 7.
Setting $G(\underline{S}) =  \widetilde{\Sigma}(h)^{-1}$ and applying Lemma \ref{lemma:stein_haff} we get 
\begin{align*}
    \mathbb{E}_{\Sigma} &\left[\mathrm{tr}\left(\Sigma^{-1}  \underline{S} \widetilde{\Sigma}(h) ^{-1} \right)\right] =  \mathbb{E}_{\Sigma} \left[ \mathrm{tr}\left\lbrace (n-p - 1) G(\underline{S}) + 2 D_{\underline{S}}(G(\underline{S})^\T \underline{S}) \right\rbrace\right] \\
    & = \mathbb{E}_{\Sigma} \left[ \mathrm{tr}\left\lbrace (n-p - 1) G(\underline{S}) + 2 D_{\underline{S}}(\widetilde{\Sigma}_h^{-1} \underline{S}) \right\rbrace\right]\\
    & = \mathbb{E}_{\Sigma} \left[ \mathrm{tr}\left\lbrace (n-p - 1) G(\underline{S}) + 2 D_{\underline{S}}(U \Delta^{-1}\Lambda^* U^\T) \right\rbrace\right]\\
    & = \mathbb{E}_{\Sigma} \left[ \mathrm{tr}\left\lbrace (n-p - 1) G(\underline{S}) + 2 D_{\underline{S}}(U \zeta(\Lambda^*) U^\T) \right\rbrace\right]\\
    & = \mathbb{E}_{\Sigma} \left[ \mathrm{tr}\left\lbrace (n-p - 1) G(\underline{S}) + 2 U \zeta^{(1)}(\Lambda^*) U^\T + \mathrm{tr}\{{\Lambda^*}^{-1} \zeta(\Lambda^*) \} (\mathrm{I_p - U U^\T})\right\rbrace \right]
\end{align*}
where $\zeta(\Lambda^*)$ is a diagonal matrix with $j$-th element $\zeta_j = \lambda_j^*/\delta_j, \, j = 1, \ldots, p$, and the last equality follows by applying Lemma \ref{lemma:D_operator}. Now, recall the definition of $\delta$:
\begin{align*}
    &\delta_{j}^{-1} = c_{1} \lambda_{j}^{-1} + c_{2} \lambda_{j}^{-1} g_n^*\left(\lambda_{j}^{-1}\right),\, \,
    g_n^*(x) = \frac{1}{p} \sum_{i=1}^p \lambda_{i}^{-1} \dfrac{\lambda_{i}^{-1} - x}{(\lambda_{i}^{-1} - x)^2 + h^2 \lambda_{i}^{-2}},  
\end{align*}
where $c_1 = (1 - p/n)$ and $c_2 = 2 (p/n)$. Hence,
\begin{align*}
    &\frac{\partial \delta_{j}^{-1}}{\partial \lambda_{j}^{-1}}  = c_{1} + c_{2} \hat{\theta}\left(\lambda_{j}^{-1}\right) + c_{2} \lambda_{j}^{-1} \frac{d g_n^*(\lambda_{j}^{-1})}{d \lambda_{j}^{-1}},\,\, 
    & \frac{d g_n^* (\lambda_{j}^{-1})}{d \lambda_{j}^{-1}}  =  \frac{1}{p} \sum_{i=1}^p \lambda_{i}^{-1} \dfrac{(\lambda_{i}^{-1} - \lambda_{j}^{-1})^2 - h^2 \lambda_{i}^{-2}}{\{(\lambda_{i}^{-1} - \lambda_{j}^{-1})^2 + h^2 \lambda_{i}^{-2}\}^2}.
\end{align*}
Therefore,
\begin{align*}
    \frac{\partial \zeta_{j}}{\partial {\lambda_{j}^*}} &= \frac{1}{\delta_j} - \frac{\lambda_j^*}{\delta_j^2}  \frac{\partial \delta_{j}}{\partial {\lambda_{j}^*}} 
 = \frac{1}{\delta_j} - \frac{\lambda_j^*}{n\delta_j^2}  \frac{\partial \delta_{j}}{\partial {\lambda_{j}}}  = \frac{1}{\delta_j} - \frac{1}{\lambda_j}  \frac{\partial \delta_{j}^{-1}}{\partial {\lambda_{j}^{-1}}}
\end{align*}
\section{Proofs of Section 3.2
}
\subsection{Proof of Theorem 8
}
The difference between the cases $m=0$ and $m>1$ comes from the presence of the term $x^m$ for $m>1$. Specifically, for $m>1$, using Lemma from \cite[Lemma 14.1]{ledoit2018optimal} we get that $\mathbf{L}_{m,n}$ has the almost sure limit
    \begin{align*}
      \mathbf{L}_m =  &\int_{-\infty}^{\infty} x^m d\Phi^{(-2)} (x) - \frac{2}{c}\sum_{k=1}^K \int_{a_k}^{b_k} \frac{x^m}{\delta(x)} \phi^{(-1)}(x) d\underline{F}(x)+ \frac{1}{c}\sum_{k=1}^K \int_{a_k}^{b_k} \frac{x^m}{\delta^2(x)} d\underline{F}(x) \\
      &+ \frac{c-1}{c} \left[ \frac{0^m}{\delta^2(0)}\phi^{(-1)}(0) - 2 \frac{\phi^{(-1)}(0) 0^m}{\delta(0)}\right]\\
      & = \int_{-\infty}^{\infty} x^m d\Phi^{(-2)} (x) - 2\sum_{k=1}^K \int_{a_k}^{b_k} \frac{x^m}{\delta(x)} \phi^{(-1)}(x) dF(x)+ \sum_{k=1}^K \int_{a_k}^{b_k} \frac{x^m}{\delta^2(x)} dF(x),
    \end{align*}
since $dF(x) = (1/c)d\underline{F}(x)$.    
For $m = 0$, $\mathbf{L}_{m, n} = \int_{-\infty}^\infty d\Phi_n^{(-2)}(x) - 2 \int_{-\infty}^{\infty} \frac{1}{\delta_n(x)}d\Phi_n^{(-1)}(x) + \int_{-\infty}^{\infty} \frac{1}{\delta_n^2(x)}dF_n(x)$. The limit can then be calculated using similar arguments as for the case $m>1$.  
\subsection{Proof of Corollary 2
}
The proof is similar to Corollary 1.
\subsection{Proof of Theorem 9
}
Recall that $\breve{m}_{\underline{F}}(x) = c\breve{m}_{F}(x) + (c-1)/x$ when $p>n$. Define $\underline{\Phi}(x) = 1 - \underline{F}(1/x)$ if $x>0$ and $0$ otherwise. Let $\underline{\Psi}(x) = \int_{-\infty}^x t d\underline{\Phi}(t)$, and for any real-valued function $g$,
$$\mathcal{H}_g(x) = \frac{1}{\pi} PV \int_{-\infty}^\infty g(t) \dfrac{1}{t-x} dt $$
denote the Hilbert transform of $g$. The following relations are true \citep[Appendix C, D]{ledoit2022quadratic}
$$\text{Re}[\breve{m}_{\underline{\Psi}}(1/x)] = - x\text{Re}[\breve{m}_{\underline{F}}(x)] \, \forall x \in \text{Supp}(\underline{F}), \,\, \text{Re}[\breve{m}_{\underline{\Psi}}(x)] = \pi \mathcal{H}_{\underline{\psi}}(x),$$
where $\psi = d\Psi$.
Then,
\begin{align*}
    \delta^*(x) &= \dfrac{x}{1 - c - 2cx \text{Re}[\breve{m}_F(x)]}\\
    & = \dfrac{x}{1 - c - 2cx \text{Re}[(1/c) \breve{m}_{\underline{F}}(x) - \{(c-1)/cx\}]}\\
    & = \dfrac{x}{c -1 - 2x \text{Re}[ \breve{m}_{\underline{F}}(x)]}\\
    & = \dfrac{x}{c -1 + 2 \text{Re}[ \breve{m}_{\underline{\Psi}}(1/x)]}\\
    & = \dfrac{x}{c -1 + 2 \pi \mathcal{H}_{\underline{\psi}}(1/x)} = \dfrac{1}{(c -1)x^{-1} + 2 \pi \mathcal{H}_{\underline{\psi}}(1/x)x^{-1}}.
\end{align*}
Next consider the shrinkage function
$$\delta_{n}^* (x) = \left[ \left(\frac{p}{n} - 1 \right) x^{-1} +   2x^{-1}  g^*_n(x^{-1}) \right]^{-1} = \left[ \left(\frac{p}{n} - 1 \right) x^{-1} +   2x^{-1}  \pi \mathcal{H}_{\underline{\psi}_n}(x^{-1})) \right]^{-1},$$
where $\underline{\Phi}_n(x) = 1 - \underline{F}_n(1/x)$,  $\underline{\Psi}_n(x) = \int_{-\infty}^x t d\underline{\Phi}_n(t)$, and $\underline{\psi}_n = d\underline{\Psi}_n$. The result follows from \cite[Theorem D.1]{ledoit2022quadratic}.

\section{Algorithm to estimate $(\pi_k, \Sigma_k)$ from Section 4
}\label{sec:algo}
    \textbf{Input:} A matrix $\widehat{B}_\star$ with $\widehat{\beta}^{(t)}_\star$ as the $t$-th row for $t = 1, \ldots, n$. 
    
    \noindent \textbf{Output:}  Mean of $\mathbb{E}(\beta^{(t)} \mid \widehat{\beta}^{(t)})$ across $T$ samples.
    \begin{enumerate}
        \item Initialize $z_1, \ldots, z_n\in \{0,1\}^{K \times 1}$, where $z_i$ is the latent indicator vector corresponding to the $i$-th row in $\widehat{B}_\star$.
        \item Set $\widehat{B}_k$ as a matrix by extracting the $i$-th row in  $\widehat{B}_\star$ 
        if and only if  $z_{ik} = 1$, and let $n_k =$ number of rows of $\widehat{B}_k$, where $z_{ik}$ is the $k$-th element in $z_i$.
        \item If $n_k > p $, estimate $\Sigma_k$ using the estimator from Theorem 6
with data $\widehat{B}_k$.
        \item If $n_k <p$, estimate $\Sigma_k$ using the estimator from Theorem 9 
        with data $\widehat{B}_k$.
        \item Estimate $\pi_k = \sum_{t=1}^n z_k/n$.
        \item Sample $z^{(t)} \sim \text{Multinomial}(p_1^{(t)}, \ldots, p_K^{(t)})$ where $p_k^{(t)} \propto \pi_k f(\widehat{\beta}^{(t)}_\star; 0, \Sigma_k) $
        \item Compute $\mathbb{E}(\beta^{(t)} \mid \widehat{\beta}^{(t)}) = \sum_{k=1}^K p_k^{(t)}(\mathrm{I} - C_k) \beta^{(t)}$ where $C_k = Q^{-1/2} \Sigma_k Q^{-1/2}$.
        \item Repeat steps 2-7 $T$ times. 
    \end{enumerate}

\newpage
\begin{landscape}
    \begin{table}[!t]
    \centering
    \scalebox{.6}{
    \begin{NiceTabular}{c|c|c||cccccccc||cccccccc}
    \toprule
    & & & \multicolumn{8}{c}{Horseshoe} & \multicolumn{8}{c}{Mixture}\\
    \midrule \addlinespace
    & & & LS  & ULS  & LLS-2 & LLS-3 & LLS-4 & UTMOST & ISA & MASH & LS  & ULS  & LLS-2 & LLS-3 & LLS-4 & UTMOST & ISA & MASH  \\
    \midrule
    \multirow{9}{*}{$n = 40$}  & \multirow{3}{*}{$p = 10$} & $\rho = 0$  & 0.005 &	0.005 &	0.001 &	0.001 &	0.003 &	0.389 &	0.222 &	41.78 & 0.005	& 0.005 &	0.00009	& 0.0004	& 0.0004 &	18.734	& 69.816	& 69.815 \\
       & &$\rho = 0.5$ & 0.009 &	0.009	& 0.002	& 0.001	& 0.002	& 0.099	& 0.467	& 11.757 & 0.009&	0.009&	0.0004 & 	0.001	& 0.001	& 22.493	& 67.585& 67.584 \\
       & & $\rho = 0.8$ & 0.019	&0.019	&0.008	& 0.009	& 0.012	& 0.646	& 0.225	& 43.664 & 0.025 &	0.025&	0.001&	0.002&	0.004	& 36.026	& 68.666 &	68.665 \\
    \cmidrule{2-19} 
      & \multirow{3}{*}{$p = 20$} & $\rho = 0$ & 0.005&	0.005&	0.001&	0.003&	0.001&	5.455&	0.291&	29.18 & 0.006 &	0.006 &	0.001 &	0.001& 	0.0004	& 15.983 & 	45.682 &	54.527 \\
       & &$\rho = 0.5$ & 0.01 &	0.01&	0.003	&0.004	&0.002	&2.262	&0.205	&21.588 & 0.01 &	0.01 &	0.001 &	0.004 &	0.005 &	24.738	& 71.846	&72.975 \\
       & & $\rho = 0.8$ & 0.024&	0.024 &	0.011 &	0.009 &	0.01&	2.616 &	0.224 &	69.25 & 0.027&	0.027 &	0.003 &	0.017 &	0.012&	41.648&	72.82&	73.879\\
      \cmidrule{2-19} 
        & \multirow{3}{*}{$p = 30$} & $\rho = 0$ & 0.006 &	0.006 &	0.002 &	0.001& 	0.0003&	0.34	&0.15	&21.305 & 0.006 &	0.006&	0.0003&	0.0002&	0.0001	&22.668&	54.901 & 65.854 \\
       & &$\rho = 0.5$ & 0.011&	0.011	& 0.002&	0.001&	0.001&	0.45&	0.149&	5.335 & 0.011	& 0.011	& 0.002 & 	0.0004 &	0.0003 & 	24.979 &	58.82	& 71.139 \\
       & & $\rho = 0.8$ & 0.027 &	0.027 &	0.014 &	0.004 &	0.015 &	0.765 &	0.155 &	9.985 & 0.028 &	0.028 &	0.003 &	0.001 &	0.001 &	37.602 &	57.827 &	72.047\\
       \bottomrule
       \multirow{9}{*}{$n = 50$}  & \multirow{3}{*}{$p = 20$} & $\rho = 0$ & 0.005	&0.005	&0.0003&	0.001&	0.001&	0.18&	0.316&	29.052 & 0.006 &	0.006 &	0.0001 &	0.014	& 0.0004&	20.947	&69.986	&69.986\\
       & &$\rho = 0.5$ & 0.01	& 0.01	&0.002	&0.002	&0.002	&0.288	&0.263	&13.965 & 0.011	&0.011	&0.001	&0.001	&0.002	&23.517	&70.569	&70.569 \\
       & & $\rho = 0.8$ & 0.023&	0.023&	0.005 &	0.008	& 0.01&	14.937&	0.201	&33.898 & 0.026	& 0.026	&0.001&	0.004 &	0.003&	41.638&	72.403	&72.403 \\
    \cmidrule{2-19} 
      & \multirow{3}{*}{$p = 30$} & $\rho = 0$ & 0.006 &	0.006 &	0.004 &	0.001&	0.0003	&2.107&	0.119&	10.879 & 0.006 &	0.006 &	0.003 &	0.001& 	0.0001	&24.988&	68.049&	74.969\\
       & &$\rho = 0.5$ & 0.007 &	0.007 &	0.005 &	0.009& 	0.005 &	3.638 &	0.004 &	8.861 & 0.011&	0.011&	0.002 &	0.001 &	0.001 &	23.57	& 62.323 &	67.754 \\
       & & $\rho = 0.8$ & 0.028 &	0.028	& 0.491	& 0.0003	&0.002&	83.803&	4.601	&227.233 & 0.029	&0.029&	0.009&	0.002&	0.002&	25.804&	62.456&	67.89\\
      \cmidrule{2-19} 
        & \multirow{3}{*}{$p = 40$} & $\rho = 0$ & 0.006 &	0.006 &	0.004 & 0.0005 &	0.0005 &	20.473 &	3.221 &	3841.965 & 0.006 &	0.006 & 	0.0002	& 0.0001 & 	0.0001 &	30.739 &	59.232 &	72.502\\
       & &$\rho = 0.5$ & 0.0121 &	0.0121	& 0.0002 &	0.0001	& 0.0004	& 10.9984	& 4.515 &	1253.4127 & 0.012 &	0.012 &	0.001 & 	0.0004 &	0.0003 &	25.182 &	51.436 &	68.82\\
       & & $\rho = 0.8$ &  0.031	& 0.031 &	0.003 &	0.0001	& 0.00009	& 7.493	& 5.558	& 238.306 & 0.03 &	0.03 &	0.002 &	0.001 &	0.001 &	36.298 &	51.744 &	69.121 \\
        
         \bottomrule
    \end{NiceTabular}
    
    }
    \caption{Mean of MSE 
    for Horseshoe and mixture settings with the same settings considered above. }
    \label{tab:MSE_T=40_p2}
\end{table}
\end{landscape}
\clearpage

\begin{landscape}
    \begin{table}
    \centering
    \scalebox{.55}{
    \begin{NiceTabular}{c|c|c||cccccccc||cccccccc}
    \toprule
    & & & \multicolumn{8}{c}{Low-rank} & \multicolumn{8}{c}{Approximately sparse}\\
    \midrule \addlinespace
    & & & LS  & ULS  & LLS-2 & LLS-3 & LLS-4 & UTMOST & ISA & MASH & LS  & ULS  & LLS-2 & LLS-3 & LLS-4 & UTMOST & ISA & MASH  \\
    \midrule
    \multirow{9}{*}{$n = 40$}  & \multirow{3}{*}{$p = 10$} & $\rho = 0$  & 1.052 &	1.052 &	1.056 &	1.06 &	1.052 &	32.243 &	35.378 &	75.498 & 1.061 &	1.061	& 1.071 &	1.074 &	1.084 &	1.195 &	1.429 &	1.06 \\
       & &$\rho = 0.5$ & 1.042	&1.042 &	1.046 &	1.046 &	1.044 &	14.891 &	2.652 &	45.399 & 1.038	&1.038 &	1.044 &	1.051 &	1.049 &	1.099 &	1.408	& 1.168\\
       & & $\rho = 0.8$ & 1.025	&1.025 &	1.037 &	1.183 &	1.03 &	16.589 &	2.394&	59.654 & 1.039	&1.039 &	1.055&	1.049 &	1.045&	1.1	& 1.405 &	1.297\\
    \cmidrule{2-19} 
      & \multirow{3}{*}{$p = 20$} & $\rho = 0$ & 1.065&	1.063 &	1.082 &	1.109 &	1.064 &	59.323 &	1.048 &	146.886 & 1.077 &	1.076 &	1.121 &	1.23	& 1.106	&1.361	&1.781 &	1.073 \\
       & &$\rho = 0.5$ & 1.052 &	1.051 &	1.116 &	1.232 &	1.063 &	38.53 &	1.058 &	159.744 & 1.094 &	1.093	&1.135 &	1.15 &	1.14 &	1.189 &	1.857 &	1.419\\
       & & $\rho = 0.8$ & 1.064	&1.064 &	1.101 &	1.075 &	1.1	&27.509 &	1.129 &	157.787 & 1.037	&1.036 &	1.068 &	1.092 &	1.088 &	1.124	& 1.728	& 1.616  \\
      \cmidrule{2-19} 
        & \multirow{3}{*}{$p = 30$} & $\rho = 0$ & 1.146 &	1.138 &	1.146 &	1.167 &	1.147 &	99.17 &	1.081 &	249.926 & 1.185 &	1.181 &	1.188 &	1.198 &	1.207 &	1.552 &	2.202 &	1.186  \\
       & &$\rho = 0.5$ & 1.135	& 1.128	& 1.135	& 167.986&	1.134 &	51.3 &	1.087 &	205.935 & 1.157&	1.153 &	1.269 &	1.176 &	1.19 &	1.265 &	2.253 &	1.828 \\
       & & $\rho = 0.8$ & 1.137&	1.136 &	1.185	& 1.137 &	1.137 &	22.861 &	1.203 &	207.818 & 1.134&	1.131 &	1.357&	1.161 &	1.467 &	1.176&	2.105&	2.07\\
       \bottomrule
       \multirow{9}{*}{$n = 50$}  & \multirow{3}{*}{$p = 20$} & $\rho = 0$ & 1.078	&1.078&	1.229&	1.092&	1.112 &	67.854&	1.062&	169.632& 1.079 &	1.078&	1.114&	1.105&	1.181&	1.326&	1.768&	1.114\\
       & &$\rho = 0.5$ & 1.068&	1.069&	1.468&	1.136&	1.068&	36.253&	1.083&	149.564 & 1.104 &	1.104&	1.125&	1.233&	1.148&	1.201&	1.787&	1.703\\
       & & $\rho = 0.8$ & 1.076&	1.075	&1.084&	1.092&	1.08	&26.82&	1.163&	176.953 & 1.09&	1.089&	1.137&	1.209&	1.128&	1.157&	1.911&	1.896\\
    \cmidrule{2-19} 
      & \multirow{3}{*}{$p = 30$} & $\rho = 0$ & 1.072&	1.069&	1.469&	1.166	&1.072&	85.227&	1.04&	217.24& 1.207 &	1.207&	1.215&	1.238&	1.212	& 1.622&	2.358&	1.277\\
       & &$\rho = 0.5$ & 1.078	&1.078&	1.078	&1.16&	1.078&	61.998&	1.05&	282.444 & 1.168 &	1.166&	1.18	&1.255&	1.169&	1.285&	2.235&	2.039\\
       & & $\rho = 0.8$ & 1.073&	1.07&	1.073&	1.073&	1.073&	28.372&	1.13&	240.155& 1.137&	1.136	&1.437&	1.22&	2.042&	1.146&	2.073&	2.057\\
      \cmidrule{2-19} 
        & \multirow{3}{*}{$p = 40$} & $\rho = 0$ & 1.196 &	1.183 &	1.196 &	1.196 &	1.196 &	126.587	& 1.087 &	308.133 & 1.291 &	1.288 &	1.287	& 1.295 &	1.267 &	1.726 &	2.621	&1.429 \\
       & &$\rho = 0.5$ & 1.239 &	1.235 &	1.239 &	1.239 &	1.239 &	62.742 &	1.189 &	325.571 & 1.28 &	1.271 &	1.973 &	1.313	& 1.336 &	1.408 &	2.714	& 2.561 \\
       & & $\rho = 0.8$ & 1.236 &	1.248	& 1.236 &	1.236 &	1.236 &	29.764 &	1.239 &	284.038 &  1.239 &	1.234 &	1.397    	& 1.266 &	1.247 &	1.288	& 2.496	& 2.493\\
        
         \bottomrule
    \end{NiceTabular}
    
    }
    \caption{Mean of PE 
    for low-rank and approximately sparse settings when the number of tissues $n = 40, 50$. For $n = 40$, the number of covariates considered are $p = 10, 20, 30$, and for $n =50$ they are $p = 20, 30, 40$. The $\rho$ denotes correlation among the covariates. }
   \label{tab:PE_T=50_p1}
\end{table}
\end{landscape}
\newpage
\begin{landscape}
    \begin{table}[!t]
    \centering
    \scalebox{.6}{
    \begin{NiceTabular}{c|c|c||cccccccc||cccccccc}
    \toprule
    & & & \multicolumn{8}{c}{Horseshoe} & \multicolumn{8}{c}{Mixture}\\
    \midrule \addlinespace
    & & & LS  & ULS  & LLS-2 & LLS-3 & LLS-4 & UTMOST & ISA & MASH & LS  & ULS  & LLS-2 & LLS-3 & LLS-4 & UTMOST & ISA & MASH  \\
    \midrule
    \multirow{9}{*}{$n = 40$}  & \multirow{3}{*}{$p = 10$} & $\rho = 0$  & 1.03	&1.03	&1.033 &	1.031 &	1.06 &	7.905 &	3.253 &	62.921 & 1.072 &	1.072 &	1.072 &	1.077 &	1.071 &	174.86 &	663.571 &	663.568  \\
       & &$\rho = 0.5$ & 1.087	&1.087	&1.088	&1.102	&1.089	&2.074	&3.477	&106.468 & 1.044 &	1.044&	1.043&	1.044&	1.045&	123.324&	704.839	&704.835\\
       & & $\rho = 0.8$ &  1.052	&1.052 &	1.051 &	1.057 &	1.063 &	3.455&	3.216 &	349.183 & 1.042 &	1.042 &	1.041 &	1.038 &	1.037 &	155.127 &	744.975 &	744.971\\
    \cmidrule{2-19} 
      & \multirow{3}{*}{$p = 20$} & $\rho = 0$ & 1.112 &	1.113 &	1.128 &	1.168 &	1.134 &	85.393&	6.364 &	320.396 & 1.107	& 1.107&	1.103 &	1.155	& 1.15	& 444.867 &	1376.433 &	1387.943\\
       & &$\rho = 0.5$ & 1.118	&1.118 &	1.139 &	1.146 &	1.129 &	57.086 &	5.168 &	507.273 & 1.131&	1.131&	1.133 &	1.19 &	1.153	&247.869 &	1341.508 &	1355.575\\
       & & $\rho = 0.8$ & 1.114	&1.114 &	1.148 &	1.136 &	1.143 &	14.528 &	5.073 &	1240.998 & 1.098 &	1.098 &	1.099 &	1.138 &	1.112&	251.947	&1518.572 &	1539.977\\
      \cmidrule{2-19} 
        & \multirow{3}{*}{$p = 30$} & $\rho = 0$ & 1.167&	1.167&	1.197&	1.188&	1.163&	15.187&	5.683&	631.574 & 1.167	& 1.166	&1.16 &	1.161&	1.16 &	680.277 &	1632.223 &	1934.98\\
       & &$\rho = 0.5$ & 1.154	&1.153 &	1.173 &	1.162 &	1.17 &	14.385 &	5.784 &	158.461 & 1.163	& 1.163 &	1.174 &	1.156 &	1.153 &	379.048&	1743.015&	1979.399  \\
       & & $\rho = 0.8$ & 1.169	&1.167	&1.239&	1.181	&1.266 &	8.398&	5.099&	270.007 & 1.19	&1.19 &	1.198&	1.18 &	1.18 &	332.021&	1780.478 &	2203.79\\
       \bottomrule
       \multirow{9}{*}{$n = 50$}  & \multirow{3}{*}{$p = 20$} & $\rho = 0$ & 1.109&	1.109	&1.11&	1.118&	1.12&	4.448&	7.113&	504.101 & 1.111&	1.111&	1.111&	1.295&	1.111&	428.534&	1447.038&	1447.033 \\
       & &$\rho = 0.5$ & 1.13&	1.13&	1.139&	1.145&	1.143&	4.844&	6.933&	304.907 & 1.106&	1.107&	1.105&	1.11&	1.117&	242.846&	1477.281&	1477.276\\
       & & $\rho = 0.8$ & 1.099&	1.099&	1.096&	1.113&	1.118&	143.691&	5.566&	711.598 & 1.068 &	1.068	&1.064&	1.069	&1.065&	253.064&	1492.048&	1492.043 \\
    \cmidrule{2-19} 
      & \multirow{3}{*}{$p = 30$} & $\rho = 0$ & 1.149&	1.15&	1.214&	1.158	&1.153&	50.395&	4.671&	358.337 & 1.142&	1.142&	1.217&	1.164&	1.137&	742.768&	2078.155&	2299.509\\
       & &$\rho = 0.5$ & 1.182	&1.181	&1.2	&1.203&	1.198&	15.029&	6.539&	121.082 & 1.152	&1.152&	1.164&	1.138&	1.135&	389.566&	2042.295&	2123.863\\
       & & $\rho = 0.8$ & 1.144&	1.144&	2.894&	1.149&	1.162&	960.661&	122.802&	3322.108 & 1.149&	1.149&	1.162&	1.129&	1.13&	160.469	&1464.253&	1678.723\\
      \cmidrule{2-19} 
        & \multirow{3}{*}{$p = 40$} & $\rho = 0$ & 1.247	&1.247 &	1.344 &	1.251 &	1.252 &	340.666 &	130.216 &	7640.402 & 1.25 &	1.249 &	1.24 &	1.229 &	1.234 &	1212.75 &	2345.017 &	2818.027\\
       & &$\rho = 0.5$ & 1.194	& 1.194 &	1.2 &	1.193 &	1.193 &	317.62	& 185.923 &	55318.149 & 1.237 &	1.237	&1.234 &	1.231 &	1.227 &	535.298 &	2061.745 &	2699.554\\
       & & $\rho = 0.8$ & 1.229 &	1.228 &	1.27 &	1.228 &	1.227 &	79.532 &	286.479&	11073.277 & 1.204 &	1.202 &	1.199 &	1.192 &	1.193	&356.081 &	1434.132 &	2051.96\\
        
         \bottomrule
    \end{NiceTabular}
    
    }
    \caption{Mean of PE 
    for Horseshoe and mixture settings with the same settings considered above.  }
    \label{tab:PE_T=50_p2}
\end{table}
\end{landscape}

\end{document}